\documentclass[12pt]{article}

\usepackage[a4paper,margin=1in]{geometry}
\usepackage{fix-cm}
\usepackage[T1]{fontenc}
\usepackage{amsmath,amssymb,amsthm,bm,mathtools}
\usepackage{newtxtext,newtxmath}
\usepackage{booktabs,multirow,array}
\usepackage{graphicx}
\usepackage{float}
\usepackage[super,sort&compress]{natbib}
\usepackage[hidelinks,hypertexnames=false]{hyperref}
\usepackage{algorithm}
\usepackage{algpseudocode}
\usepackage{setspace}
\usepackage{enumitem}

\graphicspath{{Figures/}}
\newtheorem{definition}{Definition}
\newtheorem{proposition}{Proposition}
\newtheorem{corollary}{Corollary}
\newtheorem{remark}{Remark}
\newcommand{\R}{\mathbb{R}}
\newcommand{\E}{\mathbb{E}}
\newcommand{\Var}{\operatorname{Var}}
\newcommand{\Normal}{\mathcal{N}}
\newcommand{\GIG}{\operatorname{GIG}}

\title{Mixing-Law Uncertainty in Multivariate Normal Mean-Variance Mixtures:\\
Semi-parametric Estimation and Robust Cumulative-Prospect Decisions}
\author{Nuerxiati Abudurexiti\\[0.4em]
\normalsize School of Economics and Management, Xinjiang University,\\
\normalsize Urumqi 830046, China}
\date{}

\begin{document}
\maketitle

\begin{abstract}
The distribution of a normal mean-variance mixture depends on the law of its
positive mixing variable. We compare six parametric mixing laws with a grid
nonparametric maximum likelihood estimator under the same determinant
identification constraint. The mixing mean $m=\E(Z)$ is estimated and is not
fixed at one. A paired block bootstrap is used to compare multivariate holdout
log scores. The models that cannot be distinguished from the model with the
largest score define a finite ambiguity set. We then consider a cumulative
prospect problem on a common portfolio direction. For each model in the set,
the NMVM representation gives a scalar projected return and a corresponding
prospect-value function of the exposure. The distributionally robust decision
maximizes the lower envelope of these functions. We prove existence of a
solution, give the candidate points for the piecewise smooth problem, derive a
reference-gap scaling result, and construct an interval branch-and-bound
certificate for the finite-scenario optimum. In an application to 30 stock
returns, the mixture models give higher holdout density scores than the
multivariate Gaussian model. Several parametric and semi-parametric models,
however, remain in the ambiguity set. The worst-case model is therefore
determined at the portfolio optimization stage rather than selected in advance
from a point estimate of the holdout score.
\end{abstract}


\section{Introduction}

Let the $d$-dimensional return vector $X$ satisfy
\begin{equation}\label{eq:intro-nmvm}
X=\mu+\gamma Z+\sqrt{Z}AN_d,
\qquad \Sigma=AA^\top,
\end{equation}
where $Z>0$ is independent of $N_d\sim\Normal_d(0,I_d)$. Conditional on $Z=z$,
$X$ is Gaussian with mean $\mu+z\gamma$ and covariance matrix $z\Sigma$.
Thus, $\mu$ is the location parameter, $\gamma$ determines the change in the
conditional mean with $z$, and $\Sigma$ is the conditional covariance matrix
at $z=1$. The mixing law $G=\mathcal L(Z)$ determines the random scale.
Normal inverse Gaussian, variance gamma, skewed Student-$t$, generalized
hyperbolic, and asymmetric Laplace distributions are obtained by different
choices of $G$.\cite{Aas_Haff_2006_GH_Skew_t,Adcock_2015_Skewed_in_Finance_Actuarial_science,BarndorffNielsen1997NIG,MadanCarrChang1998VG}

The choice of $G$ is part of the statistical specification. Under a
parametric specification, the parameters of $G$ and
$(\mu,\gamma,\Sigma)$ can be estimated by an EM or ECM
algorithm.\cite{Dempster_1977_EM_algorithm,Protassov2004EM} Alternatively,
the Kiefer-Wolfowitz approach estimates $G$ as a discrete probability
measure without selecting a continuous family.\cite{KieferWolfowitz1956,Laird1978NPMLE}
A parametric model and a grid estimator can fit the center of the return
distribution similarly and still imply different latent-scale tails. They
must therefore be compared under the same structural parameterization,
identification rule, and training-holdout split.

Portfolio calculations under NMVM returns usually proceed under one selected
mixing law. This is the case, for example, in likelihood-based generalized
hyperbolic models and in mean-CVaR-skewness portfolio
problems.\cite{Adcock_2015_Skewed_in_Finance_Actuarial_science,Abudurexiti_2024_Mean_CVaR_Skewness_portfolio}
If the predictive scores of several fitted laws are close, the data do not
support treating the law with the largest point score as known. This
uncertainty is relevant for nonlinear preferences because the fitted tail
probabilities enter the decision weights.

We consider this problem under cumulative prospect theory (CPT). For a
portfolio $x$, the optimization problem is
\begin{equation}\label{eq:intro-cpt-problem}
\max_{x\in\mathcal X}
V_{\rm CPT}\!\left\{r_f-r_0+x^\top(X-r_f\bm 1)\right\},
\end{equation}
where $r_f$ is the risk-free return and $r_0$ is the reference return. CPT
weights gains and losses relative to $r_0$ separately and does not reduce to
ordinary expected utility.\cite{Tversky_Kaheman_1992_CPT}
The resulting portfolio problem is generally
nonconcave.\cite{He_and_Zhou_2011_Singel_Period_Portfolio} Existing analytical
and numerical results cover several specified return laws, including
generalized hyperbolic skewed-$t$
returns.\cite{Kwak_and_Pirvu_CPT_Skew_t,Giorgio_2019_portfolio_CPT,Eric_2024_CPT_convex_optimization}
Here the mixing law is not assumed to be known after estimation.

The paper proceeds in two steps. We first estimate GIG, inverse Gaussian,
inverse gamma, gamma, exponential, and lognormal mixing models and a grid
NPMLE. All models use the same determinant normalization, and
$m=\E(Z)$ is estimated. A paired block bootstrap of the holdout log-score
differences is used to retain the models that cannot be distinguished from
the model with the largest score. These models form a finite ambiguity set.
We then fix a common portfolio direction and use the NMVM projection to obtain
one scalar return law for each retained model. The robust exposure maximizes
the lower envelope of the corresponding CPT functions. We prove existence,
derive the candidate set formed by breakpoints, stationary points, and branch
intersections, and obtain a scaling result for changes in the reference gap.
For the finite-scenario problem, an interval branch-and-bound procedure gives
a global optimality certificate. This certificate is needed because a search
based only on sign-changing roots may omit tangential stationary points or
branch intersections.

The empirical analysis uses four univariate return series for the marginal
diagnostics and a 30-asset return vector for the multivariate comparison.
Mixture models have higher multivariate holdout log scores than the Gaussian
benchmark, but five mixing specifications remain after the block-bootstrap
comparison. Their CPT functions are evaluated on the same portfolio ray, and
the active worst-case model is found at the robust optimum. The holdout sample
used to construct the ambiguity set is not used again as an independent
performance sample; the empirical CPT values reported for that period are
descriptive.

The remainder of the paper is organized as follows. Sections~\ref{sec:model}
and~\ref{sec:estimation} give the model, identification, and estimators.
Section~\ref{sec:cpt-dro} develops the common-ray CPT problem and its global
certificate. Sections~\ref{sec:empirical-design} and~\ref{sec:results} describe
the empirical design and results. Section~\ref{sec:discussion} discusses the
scope of the findings, and Section~\ref{sec:conclusion} concludes.
 \section{The NMVM model and scale identification}\label{sec:model}

\subsection{Hierarchical representation}

\begin{definition}[Normal mean-variance mixture]
Let $N_d\sim\Normal_d(0,I_d)$ and let the nonnegative random variable $Z$ be
independent of $N_d$. For $\mu,\gamma\in\R^d$ and a positive-definite matrix
$\Sigma=AA^\top$, the random vector $X$ is said to follow a $d$-dimensional
normal mean-variance mixture if
\begin{equation}\label{eq:nmvm}
X=\mu+\gamma Z+\sqrt Z\,A N_d.
\end{equation}
\end{definition}

The conditional law and marginal density are
\begin{equation}\label{eq:conditional-normal}
X\mid Z=z\sim\Normal_d(\mu+z\gamma,z\Sigma)
\end{equation}
and
\begin{equation}\label{eq:mixture-density}
f_X(x;\theta,G)=\int_0^\infty
\phi_d(x;\mu+z\gamma,z\Sigma)\,\mathrm dG(z),
\end{equation}
respectively, where $\theta=(\mu,\gamma,\Sigma)$. Whenever
$\E(Z^2)<\infty$,
\begin{align}
\E(X)&=\mu+\gamma\E(Z),\label{eq:mean}\\
\operatorname{Cov}(X)&=\E(Z)\Sigma+\Var(Z)\gamma\gamma^\top.
\label{eq:covariance}
\end{align}
It follows from Equations~\eqref{eq:mean}--\eqref{eq:covariance} that the
moments of $X$ depend on both the structural parameters and the first two
moments of $G$.

For observations $x_1,\ldots,x_n$, a parametric mixing model has observed
log-likelihood
\begin{equation}\label{eq:parametric-loglik}
\ell(\theta,\vartheta)=\sum_{i=1}^n\log
\int_0^\infty\phi_d(x_i;\mu+z\gamma,z\Sigma)
g(z;\vartheta)\,\mathrm dz,
\end{equation}
where $\vartheta$ denotes the parameters of the selected mixing density.

\subsection{Moments and scale identification}

For any $s>0$, set $Z'=Z/s$. Then
\begin{equation}\label{eq:scale-equivalence}
\gamma Z+\sqrt Z A N_d
=(s\gamma)Z'+\sqrt{Z'}(\sqrt s A)N_d.
\end{equation}
Therefore, the representation of $(G,\gamma,\Sigma)$ is not unique unless a
scale normalization is imposed. Write
\begin{equation}\label{eq:m-definition}
m=\E_G(Z),\qquad v_Z=\Var_G(Z),
\end{equation}
so that
\begin{equation}\label{eq:moments-m}
\E(X)=\mu+m\gamma,\qquad
\operatorname{Cov}(X)=m\Sigma+v_Z\gamma\gamma^\top.
\end{equation}
The quantity $m$ is not set to one. It is estimated from the selected
parametric law or from the NPMLE support weights.

Let $S_n$ be the covariance matrix of the training sample. We remove the scale
indeterminacy by imposing
\begin{equation}\label{eq:identifiability}
|\Sigma|=c_0,\qquad c_0=|S_n|.
\end{equation}
If an unconstrained update produces $\widetilde\Sigma$, define
\begin{equation}\label{eq:det-rescale}
s=\left(\frac{c_0}{|\widetilde\Sigma|}\right)^{1/d},\qquad
\Sigma=s\widetilde\Sigma,\qquad
\gamma=s\widetilde\gamma,\qquad
Z^{\rm new}=Z^{\rm old}/s.
\end{equation}
Equation~\eqref{eq:det-rescale} does not change the law of $X$. The mixing
parameters and $m$ are transformed at the same step. We use this determinant
normalization for every fitted mixing law, so their conditional covariance
matrices have the same determinant without imposing $\E(Z)=1$.\cite{McNeil2015QRM}
 \section{Parametric and semi-parametric estimation}\label{sec:estimation}

\subsection{The GIG family and its limiting models}

The generalized inverse Gaussian (GIG) density is
\begin{equation}\label{eq:gig-density}
g_{\rm GIG}(z;\lambda,\chi,\psi)=
\frac{(\psi/\chi)^{\lambda/2}}{2K_\lambda(\sqrt{\chi\psi})}
z^{\lambda-1}\exp\left[-\frac12\left(\frac\chi z+\psi z\right)\right],
\quad z>0,
\end{equation}
where $\lambda\in\R$, $\chi>0$, and $\psi>0$. Its mean and variance are
\begin{align}
m_{\rm GIG}
&=\sqrt{\frac\chi\psi}
\frac{K_{\lambda+1}(\sqrt{\chi\psi})}
{K_\lambda(\sqrt{\chi\psi})},\label{eq:gig-mean}\\
v_{\rm GIG}
&=\frac\chi\psi\left[
\frac{K_{\lambda+2}(\sqrt{\chi\psi})}
{K_\lambda(\sqrt{\chi\psi})}
-\left\{\frac{K_{\lambda+1}(\sqrt{\chi\psi})}
{K_\lambda(\sqrt{\chi\psi})}\right\}^2
\right].\label{eq:gig-var}
\end{align}
The remaining parametric mixing densities are
\begin{align}
g_{\rm Ga}(z;k,\beta)
&=\frac{\beta^k}{\Gamma(k)}z^{k-1}e^{-\beta z},
&m&=\frac{k}{\beta},&v_Z&=\frac{k}{\beta^2},
\label{eq:gamma-density}\\
g_{\rm iG}(z;\alpha,\beta)
&=\frac{\beta^\alpha}{\Gamma(\alpha)}z^{-\alpha-1}e^{-\beta/z},
\notag\\[-1mm]
m&=\frac{\beta}{\alpha-1},
&&(\alpha>1),\notag\\
v_Z&=\frac{\beta^2}{(\alpha-1)^2(\alpha-2)},
&&(\alpha>2),
\label{eq:invgamma-density}\\
g_{\rm IG}(z;m,\kappa)
&=\sqrt{\frac{\kappa}{2\pi z^3}}
\exp\left[-\frac{\kappa(z-m)^2}{2m^2z}\right],
\notag\\[-1mm]
\E(Z)&=m,
&\Var(Z)&=\frac{m^3}{\kappa},\label{eq:ig-density}\\
g_{\rm Exp}(z;\beta)&=\beta e^{-\beta z},
&m&=\frac1\beta,&v_Z&=\frac1{\beta^2},
\label{eq:exp-density}\\
g_{\rm LN}(z;\eta,\tau)
&=\frac1{z\tau\sqrt{2\pi}}
\exp\left[-\frac{(\log z-\eta)^2}{2\tau^2}\right],
\notag\\[-1mm]
m&=e^{\eta+\tau^2/2},
&v_Z&=(e^{\tau^2}-1)e^{2\eta+\tau^2}.
\label{eq:lognormal-density}
\end{align}

The inverse-gamma parameterization is the one used in the skewed-$t$
representation. If
$Z_0\sim\operatorname{InvGamma}(\nu/2,\nu/2)$ and the determinant
normalization rescales $Z=Z_0/s$, then
\begin{equation}\label{eq:skewt-invgamma-map}
\alpha=\frac\nu2,\qquad
\beta=\frac\nu{2s},\qquad
\gamma=s\gamma_0,\qquad
\Sigma=s\Sigma_0,
\end{equation}
and $m=\beta/(\alpha-1)$. A finite first moment only requires $\alpha>1$.
For the variance comparison below, the inverse-gamma estimation is restricted
to $\alpha\geq2.05$. This restriction places the estimate above the
second-moment boundary $\alpha=2$. An estimate equal to 2.05 is therefore
reported as a constrained estimate.

Gamma is obtained from GIG as $\chi\downarrow0$, $\lambda=k$, and
$\psi=2\beta$; inverse gamma corresponds to $\psi\downarrow0$,
$\lambda=-\alpha$, and $\chi=2\beta$; inverse Gaussian corresponds to
$\lambda=-1/2$, $\chi=\kappa$, and $\psi=\kappa/m^2$; and the exponential law
is gamma with $k=1$. The lognormal mixing law is not a GIG member.

Let
\begin{equation}\label{eq:quadratic-forms-density}
\delta(x)=(x-\mu)^\top\Sigma^{-1}(x-\mu),
\qquad q=\gamma^\top\Sigma^{-1}\gamma.
\end{equation}
Substitution of Equation~\eqref{eq:gig-density} into
Equation~\eqref{eq:mixture-density}, followed by the standard Bessel integral,
gives
\begin{align}
f_{\rm GIG}(x)
={}&\frac{\exp\{(x-\mu)^\top\Sigma^{-1}\gamma\}}
{(2\pi)^{d/2}|\Sigma|^{1/2}}
\frac{(\psi/\chi)^{\lambda/2}}
{K_\lambda(\sqrt{\chi\psi})}\notag\\
&\times
\left(\frac{\chi+\delta(x)}{\psi+q}\right)^{(\lambda-d/2)/2}
K_{\lambda-d/2}
\!\left(\sqrt{\{\chi+\delta(x)\}(\psi+q)}\right).
\label{eq:gh-density}
\end{align}
The corresponding limits of Equation~\eqref{eq:gh-density} are used for the
four GIG boundary models. The lognormal mixture is evaluated numerically in
Equation~\eqref{eq:parametric-loglik}.

\begin{table}[htbp]
\centering
\caption{Mixing laws and the corresponding NMVM submodels.}
\label{tab:family-map}
\begin{tabular}{lll}
\toprule
Mixing law $G$ & Parameter restriction & Common name for $X$\\
\midrule
GIG & $(\lambda,\chi,\psi)$ & generalized hyperbolic (GH)\\
inverse Gaussian & $\lambda=-1/2$ & normal inverse Gaussian (NIG)\\
gamma & $\chi=0,\lambda>0$ & variance gamma (VG)\\
inverse gamma & $\psi=0,\lambda<0$ & skewed Student-$t$\\
exponential & gamma with $k=1$ & asymmetric Laplace\\
lognormal & outside GIG & lognormal scale mixture\\
\bottomrule
\end{tabular}
\end{table}

\subsection{EM/ECM updates for GIG mixtures}

For observation $x_i$, define
\begin{equation}\label{eq:quadratic-terms}
\delta_i=(x_i-\mu)^\top\Sigma^{-1}(x_i-\mu),
\qquad q=\gamma^\top\Sigma^{-1}\gamma.
\end{equation}
Under GIG mixing, conjugacy gives
\begin{equation}\label{eq:posterior-gig}
Z_i\mid X_i=x_i\sim
\GIG\left(\lambda-\frac d2,\chi+\delta_i,\psi+q\right).
\end{equation}
If $W\sim\GIG(\ell,u,v)$, then
\begin{align}
\E(W^r)&=\left(\frac uv\right)^{r/2}
\frac{K_{\ell+r}(\sqrt{uv})}{K_\ell(\sqrt{uv})},
\label{eq:gig-moment}\\
\E(\log W)&=\frac12\log\frac uv+
\frac{\partial}{\partial\ell}\log K_\ell(\sqrt{uv}).
\label{eq:gig-logmoment}
\end{align}
The E-step computes
\begin{equation}\label{eq:abc}
a_i=\E(Z_i\mid x_i),\qquad
b_i=\E(Z_i^{-1}\mid x_i),\qquad
c_i=\E(\log Z_i\mid x_i).
\end{equation}
Write $A_1=\sum_i a_i$, $B_1=\sum_i b_i$, $S=\sum_i x_i$, and
$T=\sum_i b_i x_i$. The part of the conditional objective involving the
structural parameters is
\begin{align}
Q_{\theta}
={}&-\frac n2\log|\Sigma|-\frac12\sum_{i=1}^n
\left[b_i(x_i-\mu)^\top\Sigma^{-1}(x_i-\mu)\right.\notag\\
&\left.\hspace{18mm}-2(x_i-\mu)^\top\Sigma^{-1}\gamma
+a_i\gamma^\top\Sigma^{-1}\gamma\right]+C,
\label{eq:structural-Q}
\end{align}
where $C$ is constant in $\theta$. Setting the derivatives with respect to
$\mu$, $\gamma$, and $\Sigma$ to zero yields
\begin{align}
\widehat\mu&=\frac{A_1T-nS}{A_1B_1-n^2},
\label{eq:mu-mv-update}\\
\widehat\gamma&=\frac{B_1S-nT}{A_1B_1-n^2},
\label{eq:gamma-mv-update}\\
\widehat\Sigma&=\frac1n\sum_{i=1}^n
\left[b_i r_i r_i^\top-r_i\widehat\gamma^\top
-\widehat\gamma r_i^\top+a_i\widehat\gamma\widehat\gamma^\top\right],
\quad r_i=x_i-\widehat\mu.
\label{eq:sigma-update}
\end{align}
The GIG parameters maximize
\begin{align}
Q_{\rm GIG}={}&n\left[\frac\lambda2(\log\psi-\log\chi)
-\log\{2K_\lambda(\sqrt{\chi\psi})\}\right]
+(\lambda-1)\sum_i c_i\notag\\
&-\frac\chi2\sum_i b_i-\frac\psi2\sum_i a_i.
\label{eq:gig-Q}
\end{align}
We perform bounded numerical maximization over
$(\lambda,\log\chi,\log\psi)$. For inverse Gaussian mixing,
$\lambda=-1/2$ is fixed and only $\chi$ and $\psi$ are updated; the reported
parameters are $\widehat\kappa=\widehat\chi$ and
$\widehat m=(\widehat\chi/\widehat\psi)^{1/2}$.

\subsection{Boundary families and lognormal mixing}

For gamma mixing, the part of the complete-data objective involving
$(k,\beta)$ is
\begin{equation}\label{eq:gamma-Q}
Q_{\rm Ga}=n\{k\log\beta-\log\Gamma(k)\}
+(k-1)\sum_i c_i-\beta\sum_i a_i.
\end{equation}
Hence
\begin{equation}\label{eq:gamma-update}
\widehat\beta=\frac{\widehat k}{\bar a},\qquad
\log\widehat\beta-\psi_0(\widehat k)+\bar c=0,
\end{equation}
where $\bar a=n^{-1}\sum_i a_i$, $\bar c=n^{-1}\sum_i c_i$, and $\psi_0$
is the digamma function. Exponential mixing fixes $k=1$, giving
$\widehat\beta=1/\bar a$. The posterior distributions needed in the E-step are
\begin{align}
Z_i\mid x_i,{\rm Gamma}
&\sim\GIG\left(k-\frac d2,\delta_i,2\beta+q\right),\notag\\
Z_i\mid x_i,{\rm Exp}
&\sim\GIG\left(1-\frac d2,\delta_i,2\beta+q\right),\notag\\
Z_i\mid x_i,{\rm IG}
&\sim\GIG\left(-\frac12-\frac d2,\kappa+\delta_i,
\frac\kappa{m^2}+q\right).
\label{eq:restricted-posteriors}
\end{align}
Their conditional moments follow directly from
Equations~\eqref{eq:gig-moment}--\eqref{eq:gig-logmoment}.

For inverse-gamma mixing,
\begin{equation}\label{eq:invgamma-posterior}
Z_i\mid x_i\sim
\GIG\left(-\alpha-\frac d2,2\beta+\delta_i,q\right),
\end{equation}
and
\begin{equation}\label{eq:invgamma-Q}
Q_{\rm iG}=n\{\alpha\log\beta-\log\Gamma(\alpha)\}
-(\alpha+1)\sum_i c_i-\beta\sum_i b_i.
\end{equation}
The update solves
\begin{equation}\label{eq:invgamma-update}
\widehat\beta=\frac{\widehat\alpha}{\bar b},\qquad
\log\widehat\beta-\psi_0(\widehat\alpha)-\bar c=0,
\qquad \widehat\alpha\geq2.05,
\end{equation}
where $\bar b=n^{-1}\sum_i b_i$. The skewed-$t$ degrees-of-freedom
parameter is recovered as $\widehat\nu=2\widehat\alpha$; the reported
$\widehat\beta$ is the scale after determinant normalization.

The lognormal posterior is not available in closed form. Let $(u_l,w_l)$,
$l=1,\ldots,L_Q$, be Gauss-Hermite nodes and weights, and set
\begin{equation}\label{eq:gh-nodes}
z_l=\exp(\eta+\sqrt2\tau u_l),\qquad \omega_l=w_l/\sqrt\pi.
\end{equation}
The quadrature responsibilities and conditional moments are
\begin{align}
\tau_{il}&=\frac{\omega_l\phi_d(x_i;\mu+z_l\gamma,z_l\Sigma)}
{\sum_h\omega_h\phi_d(x_i;\mu+z_h\gamma,z_h\Sigma)},
\label{eq:ln-responsibility}\\
a_i&=\sum_l\tau_{il}z_l,\quad
b_i=\sum_l\tau_{il}z_l^{-1},\quad
c_i=\sum_l\tau_{il}\log z_l,\quad
d_i=\sum_l\tau_{il}(\log z_l)^2.
\label{eq:ln-moments}
\end{align}
The lognormal parameter updates are
\begin{equation}\label{eq:ln-update}
\widehat\eta=\frac1n\sum_i c_i,
\qquad
\widehat\tau^2=\frac1n\sum_i d_i-\widehat\eta^2.
\end{equation}

After each structural update, Equation~\eqref{eq:det-rescale} is applied. The
corresponding changes in the mixing parameters are
\begin{align}
&(\lambda,\chi,\psi)\leftarrow(\lambda,\chi/s,s\psi),\notag\\
&(m_{\rm IG},\kappa)\leftarrow(m_{\rm IG}/s,\kappa/s),\notag\\
&(k,\beta_{\rm Ga})\leftarrow(k,s\beta_{\rm Ga}),\qquad
\beta_{\rm Exp}\leftarrow s\beta_{\rm Exp},\notag\\
&(\alpha,\beta_{\rm iG})\leftarrow(\alpha,\beta_{\rm iG}/s),\notag\\
&(\eta,\tau)\leftarrow(\eta-\log s,\tau).
\label{eq:family-rescale}
\end{align}
These transformations update $m$ but leave the observed-data density
unchanged.

\begin{algorithm}[htbp]
\caption{EM/ECM estimation for a parametric mixing law}
\begin{algorithmic}[1]
\State Select GIG, inverse Gaussian, gamma, inverse gamma, exponential, or lognormal mixing, and initialize from sample moments.
\Repeat
  \State Compute $a_i,b_i,c_i$ from the relevant posterior, or from Equations~\eqref{eq:gh-nodes}--\eqref{eq:ln-moments}.
  \State Update $\mu,\gamma,\Sigma$ by Equations~\eqref{eq:mu-mv-update}--\eqref{eq:sigma-update}.
  \State Update the mixing parameters and normalize the scale using Equations~\eqref{eq:det-rescale} and~\eqref{eq:family-rescale}.
\Until{the relative increase in observed log-likelihood is below the prescribed tolerance.}
\State Return the converged parameter estimates.
\end{algorithmic}
\end{algorithm}

\subsection{Grid NPMLE for an unknown mixing law}

The semi-parametric estimator places no continuous parametric form on $G$.
Given a positive grid $0<z_1<\cdots<z_K$, write
\begin{equation}\label{eq:discrete-G}
G_K=\sum_{j=1}^Kp_j\delta_{z_j},
\qquad p_j\geq0,\quad\sum_{j=1}^Kp_j=1.
\end{equation}
Its log-likelihood is
\begin{equation}\label{eq:npmle-loglik}
\ell_{sp}(\theta,p)=
\sum_{i=1}^n\log\left[
\sum_{j=1}^Kp_j\phi_d(x_i;\mu+z_j\gamma,z_j\Sigma)
\right].
\end{equation}
The NPMLE of a mixing distribution may be taken to have finite support.\cite{KieferWolfowitz1956,Laird1978NPMLE}
Conditional on the grid, the E-step responsibilities are
\begin{equation}\label{eq:responsibility}
\tau_{ij}=
\frac{p_j\phi_d(x_i;\mu+z_j\gamma,z_j\Sigma)}
{\sum_{h=1}^Kp_h\phi_d(x_i;\mu+z_h\gamma,z_h\Sigma)}.
\end{equation}
The M-step uses
\begin{equation}\label{eq:npmle-moments}
p_j=\frac1n\sum_i\tau_{ij},\qquad
a_i=\sum_j\tau_{ij}z_j,
\qquad b_i=\sum_j\tau_{ij}z_j^{-1},
\end{equation}
followed by Equations~\eqref{eq:mu-mv-update}--\eqref{eq:sigma-update}.
During iteration, the implementation uses the equivalent coordinate
$\sum_jp_jz_j=1$ for numerical conditioning. After termination, $s$ is computed
from Equation~\eqref{eq:det-rescale}, and the estimate is mapped to the
determinant-constrained representative by
\begin{equation}\label{eq:grid-normalize}
z_j\leftarrow z_j/s,\qquad
\gamma\leftarrow s\gamma,\qquad
\Sigma\leftarrow s\Sigma,\qquad
\widehat m=\sum_jp_jz_j,\qquad
\widehat v_Z=\sum_jp_j(z_j-\widehat m)^2.
\end{equation}
Equation~\eqref{eq:grid-normalize} leaves the observed-data density unchanged.
Thus, the mean-one coordinate is used only during iteration, and the reported
$\widehat m$ is generally different from one.
The initial grid contains $K=45$ support points, logarithmically spaced on
$[0.025,16]$. Responsibilities are evaluated in the log domain, and a support
point is called effective when $p_j>10^{-3}$.

\subsection{Parameter counts and numerical comparison}

The determinant constraint removes one degree of freedom from $\Sigma$. Let
\begin{equation}\label{eq:base-parameter-count}
p_0=2d+\frac{d(d+1)}2-1.
\end{equation}
The numbers of free parameters for GIG, inverse Gaussian, gamma, inverse
gamma, exponential, and lognormal mixing are, respectively,
$p_0+3$, $p_0+2$, $p_0+2$, $p_0+2$, $p_0+1$, and $p_0+2$.
The effective dimension of the NPMLE depends on the number of positive
support weights. We compare the parametric fits by AIC, BIC, and the holdout
mean log score. The holdout score is the principal criterion for the
semi-parametric fit.
The information criteria are
\begin{equation}\label{eq:aic-bic}
\mathrm{AIC}=-2\widehat\ell+2p,
\qquad
\mathrm{BIC}=-2\widehat\ell+p\log n,
\end{equation}
where $p$ is the number of free parameters after scale identification.
 \section{From mixing-law uncertainty to robust CPT decisions}\label{sec:cpt-dro}

After estimating $G$, $\mu$, $\gamma$, and $\Sigma$, we consider the portfolio
return relative to a reference point. We first express an arbitrary
$d$-asset portfolio by three scalar projection parameters. We then fix a
common portfolio direction and denote its exposure by $c$. Each fitted model
gives one CPT function of $c$, and the pointwise minimum of these functions is
the distributionally robust objective.

\subsection{Reference point, value function, and probability weighting}

Let $Y$ denote a one-period portfolio return relative to the reference return
$r_0$. The gain and loss value functions are
\begin{equation}\label{eq:cpt-value-functions}
u_+(y)=y^{\alpha_+},\qquad
u_-(y)=\lambda_{\ell}y^{\alpha_-},\qquad y\geq0,
\end{equation}
where $0<\alpha_+,\alpha_-\leq1$ and $\lambda_{\ell}>1$ measures loss
aversion. Probability weighting on the gain and loss sides is
\begin{equation}\label{eq:cpt-weight-functions}
w^{\pm}(p)=
\frac{p^{\delta_{\pm}}}
{\left[p^{\delta_{\pm}}+(1-p)^{\delta_{\pm}}\right]^{1/\delta_{\pm}}},
\qquad 0\leq p\leq1.
\end{equation}
If $F_Y$ is continuous, the cumulative prospect value can be written as the
Lebesgue-Stieltjes integral
\begin{align}
V_{\rm CPT}(Y)
={}&\int_0^{\infty}u_+(y)\,
\mathrm d\{-w^+[1-F_Y(y)]\}\notag\\
&-\int_{-\infty}^{0}u_-(-y)\,
\mathrm d\{w^-[F_Y(y)]\}.
\label{eq:cpt-functional}
\end{align}
The first integral uses the distorted survival probability of gains, whereas
the second uses the distorted distribution probability of losses. Both $r_f$
and $r_0$ are deterministic returns measured at the same frequency as $X$.

\subsection{Exact three-parameter representation of a portfolio}

Let $x$ be the vector of risky-asset weights, with the remaining wealth held
in the risk-free asset. The portfolio return relative to $r_0$ is
\begin{equation}\label{eq:portfolio-relative-return}
Y(x)=r_f-r_0+x^\top(X-r_f\bm1).
\end{equation}
For $\Sigma=AA^\top$, define
\begin{equation}\label{eq:whitened-directions}
y=A^\top x,\qquad
\mu_0=A^{-1}(\mu-r_f\bm1),\qquad
\gamma_0=A^{-1}\gamma.
\end{equation}
The NMVM representation gives
\begin{equation}\label{eq:portfolio-nmvm-projection}
x^\top(X-r_f\bm1)
\overset{\mathrm d}=
y^\top\mu_0+y^\top\gamma_0Z+\|y\|\sqrt Z\,N,
\qquad N\sim\Normal(0,1),\quad N\perp Z.
\end{equation}
For $x\neq0$ and $\|\mu_0\|\,\|\gamma_0\|>0$, let
\begin{equation}\label{eq:three-projection-parameters}
\rho=\|y\|=\sqrt{x^\top\Sigma x},\qquad
\phi=\frac{y^\top\gamma_0}{\rho\|\gamma_0\|},\qquad
\psi=\frac{y^\top\mu_0}{\rho\|\mu_0\|}.
\end{equation}
If $\mu_0$ or $\gamma_0$ is zero, the associated directional cosine is set to
zero and the corresponding term is omitted. Thus
\begin{equation}\label{eq:three-parameter-law}
Y(x)\overset{\mathrm d}=
r_f-r_0+\rho\left(
\|\mu_0\|\psi+\|\gamma_0\|\phi Z+\sqrt Z\,N
\right).
\end{equation}

\begin{proposition}[Exact NMVM projection of the CPT objective]
\label{prop:cpt-projection}
Let $\mathcal X$ be the feasible set of portfolio weights and suppose that
$V_{\rm CPT}$ depends on $x$ only through the law of $Y(x)$. Define
\begin{equation}\label{eq:projection-image-set}
\mathcal D_{\mathcal X}=
\left\{(\rho(x),\phi(x),\psi(x)):x\in\mathcal X,\ x\neq0\right\},
\end{equation}
with $x=0$ treated separately. Maximizing $V_{\rm CPT}\{Y(x)\}$ over
$\mathcal X$ is equivalent to maximizing the CPT value of
Equation~\eqref{eq:three-parameter-law} over $\mathcal D_{\mathcal X}$ and
comparing the result with $V_{\rm CPT}(r_f-r_0)$.
\end{proposition}

\begin{proof}
Equation~\eqref{eq:portfolio-nmvm-projection} follows from the distribution of
a linear combination of a Gaussian vector. Equation~\eqref{eq:three-parameter-law}
rewrites the three inner products in terms of lengths and directional cosines.
Any two portfolios with the same $(\rho,\phi,\psi)$ induce the same law of $Y$
and therefore the same CPT value. Taking the image of the feasible set proves
the claim.
\end{proof}

No concavity assumption or parametric law for $Z$ is used in the proposition.
It therefore applies to the parametric estimates and the grid NPMLE. The three
coordinates must, however, be generated by the same portfolio $x$; an
arbitrary rectangular set of values for $(\rho,\phi,\psi)$ is not equivalent
to $\mathcal D_{\mathcal X}$.

\subsection{From a model-specific fund direction to a common decision ray}

For hyperbolic returns and a continuous concave utility function, an optimal
risky-fund direction is proportional to
$\Sigma^{-1}\{\mu-r_f\bm1+\gamma\E(Z)\}$.\cite{Abudurexiti_2024_expected_utility}
CPT does not satisfy the concave expected-utility assumptions of that result.
We use this direction to define a one-dimensional feasible set. The resulting
solution is an optimum on that set, not an optimum of the unrestricted CPT
problem.

For one fitted NMVM law $\widehat P$, set
\begin{equation}\label{eq:common-fund-direction}
\widehat v=\widehat\mu-r_f\bm1+\widehat m\widehat\gamma,
\qquad
\widehat D=\widehat v^\top\widehat\Sigma^{-1}\widehat v,
\qquad
\widehat q=
\frac{\widehat\Sigma^{-1}\widehat v}{\widehat D}.
\end{equation}
Because $\widehat q^\top\widehat v=1$, the random excess return along
$x(c)=c\widehat q$ is
\begin{equation}\label{eq:eta-definition}
x(c)^\top(X-r_f\bm1)=c\eta,
\end{equation}
where
\begin{equation}\label{eq:eta-nmvm-law}
\eta=
\frac{\widehat v^\top\widehat\Sigma^{-1}
(\widehat\mu-r_f\bm1)}{\widehat D}
+\frac{\widehat v^\top\widehat\Sigma^{-1}\widehat\gamma}
{\widehat D}Z
+\frac{\sqrt Z}{\sqrt{\widehat D}}N,
\qquad \E(\eta)=1.
\end{equation}
The last equality follows from the normalization of the fund direction:
\begin{equation}\label{eq:eta-mean-normalization}
\E(\eta)
=\widehat q^\top(\widehat\mu-r_f\bm1)
+\widehat m\,\widehat q^\top\widehat\gamma
=\widehat q^\top\widehat v=1.
\end{equation}
This equality does not impose $\E(Z)=1$; it uses the fitted value
$\widehat m=\E_{\widehat G}(Z)$. The division by $\widehat D$ normalizes the
expected excess return of the fund to one. If instead
$\widetilde q=\widehat\Sigma^{-1}\widehat v=\widehat D\widehat q$ and
$\widetilde\eta=\widetilde q^\top(X-r_f\bm1)$, then
$\widetilde\eta=\widehat D\eta$ and $\E(\widetilde\eta)=\widehat D$. Hence
\begin{equation}\label{eq:fund-ray-rescaling}
c\widehat q=\widetilde c\,\widetilde q,
\qquad
\widetilde c=\frac{c}{\widehat D},
\qquad
\widetilde c\,\widetilde\eta=c\eta.
\end{equation}
Thus, the two parameterizations give the same portfolio and the same random
return. Only the units of the exposure coefficient differ.

Under the gross-exposure bound $\|x\|_1\leq L$ and $c\geq0$,
\begin{equation}\label{eq:c-upper-bound}
0\leq c\leq c_{\max},
\qquad
c_{\max}=\frac{L}{\|\widehat q\|_1}.
\end{equation}
The nominal ray problem is
\begin{equation}\label{eq:c-star-problem}
\widehat c^*\in
\arg\max_{0\leq c\leq c_{\max}}
V_{\rm CPT}(r_f-r_0+c\eta),
\qquad
\widehat x^*=\widehat c^*\widehat q.
\end{equation}
Equation~\eqref{eq:c-star-problem} is the optimization problem on the stated
ray. If the direction is also a decision variable, one must optimize over the
three-parameter image in Proposition~\ref{prop:cpt-projection}.

Suppose now that $K$ models survive the multivariate predictive comparison.
The finite ambiguity set is
\begin{equation}\label{eq:finite-ambiguity-set}
\widehat{\mathfrak P}=\{\widehat P_1,\ldots,\widehat P_K\}.
\end{equation}
To compare the candidate laws, we evaluate the same portfolio under every
law. Fix a nonzero direction $q_0$ that does not depend on $k$, and define
\begin{equation}\label{eq:robust-common-ray}
\mathcal X(q_0)=\{x(c)=cq_0:0\leq c\leq c_{\max}\},
\qquad
c_{\max}=\frac{L}{\|q_0\|_1}.
\end{equation}
The direction can be constructed from training-sample moments or from a fixed
reference model. It is not re-estimated for the individual elements of
$\widehat{\mathfrak P}$. When $K=1$, the choice $q_0=\widehat q$ gives
Equation~\eqref{eq:c-star-problem}. When $K>1$, the projected means can differ
even though $q_0$ is common.

For $X^{(k)}\sim\widehat P_k$, define
\begin{equation}\label{eq:robust-eta-branches}
\eta_k=q_0^\top(X^{(k)}-r_f\bm1),
\qquad
\Gamma_k(c)=V_{\rm CPT}^{\widehat P_k}
\!\left(r_f-r_0+c\eta_k\right).
\end{equation}
If $\widehat m_k=\E_{\widehat P_k}(Z)$, then
\begin{equation}\label{eq:robust-eta-means}
\E_{\widehat P_k}(\eta_k)
=q_0^\top\left(
\widehat\mu_k-r_f\bm1+\widehat m_k\widehat\gamma_k
\right).
\end{equation}
The equality $\E(\eta)=1$ holds when $q_0$ is the normalized direction of the
same fitted model. Equation~\eqref{eq:robust-eta-means} is generally different
from one for the other laws. Their projected means, skewness loadings, and
conditional scales enter $\Gamma_k(c)$, while $c$ refers to the same portfolio
for every $k$. Define
\begin{equation}\label{eq:dro-lower-envelope}
\underline\Gamma(c)
=\inf_{P\in\widehat{\mathfrak P}}
V_{\rm CPT}^{P}
\!\left(r_f-r_0+cq_0^\top(X-r_f\bm1)\right)
=\min_{1\leq k\leq K}\Gamma_k(c).
\end{equation}

\begin{proposition}[Robust reduction under a finite distribution set]
\label{prop:finite-distribution-dro}
Assume $0<c_{\max}<\infty$ and that every $\Gamma_k$ is finite and continuous
on $[0,c_{\max}]$. Then the distributionally robust ray problem is
\begin{equation}\label{eq:dro-c-problem}
c_{\rm DR}^*\in
\arg\max_{0\leq c\leq c_{\max}}\underline\Gamma(c)
=\arg\max_{0\leq c\leq c_{\max}}
\min_{1\leq k\leq K}\Gamma_k(c),
\qquad
x_{\rm DR}^*=c_{\rm DR}^*q_0,
\end{equation}
and its solution set is nonempty. At a fixed $c$, the set of worst-case models
is
\begin{equation}\label{eq:active-worst-distribution}
\mathcal I(c)=\arg\min_{1\leq k\leq K}\Gamma_k(c).
\end{equation}

Suppose further that each $\Gamma_k$ is continuously differentiable except at
a finite set $\mathcal B_k\subset(0,c_{\max})$, and let
$\mathcal B=\bigcup_{k=1}^K\mathcal B_k$. At least one robust optimum belongs
to
\begin{align}
\mathcal C_{\rm DR}={}&\{0,c_{\max}\}\cup\mathcal B\notag\\
&\cup\bigcup_{k=1}^K
\left\{c\in(0,c_{\max})\setminus\mathcal B:
\Gamma_k'(c)=0\right\}\notag\\
&\cup\bigcup_{1\leq i<j\leq K}
\left\{c\in(0,c_{\max}):\Gamma_i(c)=\Gamma_j(c)\right\}.
\label{eq:dro-candidate-set}
\end{align}
If the stationary points and intersections in
Equation~\eqref{eq:dro-candidate-set} are finite, the robust problem reduces to
pointwise comparison over the finite set $\mathcal C_{\rm DR}$.
\end{proposition}

\begin{proof}
Equation~\eqref{eq:robust-common-ray} makes $x(c)=cq_0$ the same decision under
all candidate laws, leaving only the scalar $c$. Since
$\widehat{\mathfrak P}$ is finite, for every fixed $c$,
\[
\inf_{P\in\widehat{\mathfrak P}}V_{\rm CPT}^{P}
\!\left(r_f-r_0+cq_0^\top(X-r_f\bm1)\right)
=\min_{1\leq k\leq K}\Gamma_k(c),
\]
and the minimum is attained by at least one $k$, which proves
Equation~\eqref{eq:active-worst-distribution}.

The pointwise minimum of finitely many continuous functions is continuous. In
particular, for $c_1,c_2\in[0,c_{\max}]$,
\begin{equation}\label{eq:min-continuity-bound}
|\underline\Gamma(c_1)-\underline\Gamma(c_2)|
\leq\max_{1\leq k\leq K}
|\Gamma_k(c_1)-\Gamma_k(c_2)|.
\end{equation}
The extreme-value theorem then gives existence on the compact interval.

Let $c^\dagger$ be a global maximizer. If it is an endpoint or belongs to
$\mathcal B$, it already appears in Equation~\eqref{eq:dro-candidate-set}.
Otherwise, it is an interior point at which every branch is differentiable.
If $|\mathcal I(c^\dagger)|\geq2$, two active branches agree there, so
$c^\dagger$ is a branch intersection. If
$\mathcal I(c^\dagger)=\{k^\dagger\}$, finiteness and continuity imply that
$\underline\Gamma=\Gamma_{k^\dagger}$ in a neighborhood of $c^\dagger$.
The first-order necessary condition for an interior extremum gives
$\Gamma_{k^\dagger}'(c^\dagger)=0$. Hence at least one global maximizer is in
$\mathcal C_{\rm DR}$.
\end{proof}

\begin{remark}[A verifiable sufficient condition for continuity]
\label{rem:dro-continuity-condition}
The continuity assumption in Proposition~\ref{prop:finite-distribution-dro}
can be checked from projected moments. If for each $k$ there is a $q_k$ such
that
\begin{equation}\label{eq:dro-moment-condition}
q_k>\max\left\{\frac{\alpha_+}{\delta_+},
\frac{\alpha_-}{\delta_-}\right\},
\qquad
\E_{\widehat P_k}|\eta_k|^{q_k}<\infty,
\end{equation}
and if $\widehat\Sigma_k$ is positive definite and $Z>0$, then $\Gamma_k$ is
finite and continuous on $[0,c_{\max}]$. Appendix~\ref{app:dro-continuity-proof}
gives the tail-integral and dominated-convergence argument.
\end{remark}

\begin{remark}
Proposition~\ref{prop:finite-distribution-dro} is stated for the finite set
$\widehat{\mathfrak P}$, not for its convex hull. Probability weighting makes
$V_{\rm CPT}^{P}$ nonlinear in $P$ in general. Without additional conditions,
$\inf_{P\in\operatorname{co}(\widehat{\mathfrak P})}V_{\rm CPT}^{P}$ cannot
be replaced by the minimum of its values at the listed models. Also,
$\min_k\Gamma_k(c;q_k)$ with a model-specific direction $q_k$ compares
different portfolios and is not the worst-case value of one decision.
\end{remark}

\subsection{Homogeneity of the lower envelope and reference-gap scaling}

Suppose first that $\alpha_+=\alpha_-=\alpha$. For any random variable $W$
and scalar $a>0$, positive rescaling preserves signs, ranks, and decision
weights. Hence
\begin{equation}\label{eq:cpt-positive-homogeneity}
V_{\rm CPT}(aW)=a^\alpha V_{\rm CPT}(W).
\end{equation}
In particular, if $r_0=r_f$, then for every $k$,
\begin{equation}\label{eq:cpt-homogeneity}
\Gamma_k(c)=c^\alpha\Gamma_k(1),
\qquad
\underline\Gamma(c)=c^\alpha\underline\Gamma(1),
\qquad c\geq0.
\end{equation}
It follows that the robust exposure is $c_{\max}$ if
$\underline\Gamma(1)>0$, zero if $\underline\Gamma(1)<0$, and every
$c\in[0,c_{\max}]$ is optimal if $\underline\Gamma(1)=0$.

When the reference return exceeds the risk-free return, define
\begin{equation}\label{eq:reference-gap-standardization}
h=r_0-r_f>0,\qquad t=\frac ch.
\end{equation}
For each fitted law,
\begin{equation}\label{eq:cpt-reference-gap-homogeneity}
V_{\rm CPT}^{\widehat P_k}(-h+c\eta_k)
=h^\alpha V_{\rm CPT}^{\widehat P_k}
\!\left(-1+\frac ch\eta_k\right).
\end{equation}
The factor $h^\alpha$ is common and strictly positive, so it also factors out
of the pointwise minimum.

\begin{corollary}[Scaling of robust exposure with the reference gap]
\label{cor:dro-gap-scaling}
Let $h=r_0-r_f>0$ and $\alpha_+=\alpha_-=\alpha$. Define
\begin{equation}\label{eq:standardized-drobjective}
\underline G(t)=\min_{1\leq k\leq K}
V_{\rm CPT}^{\widehat P_k}(-1+t\eta_k),
\qquad t\geq0.
\end{equation}
Then
\begin{equation}\label{eq:dro-gap-homogeneity}
\underline\Gamma_h(c)
:=\min_{1\leq k\leq K}
V_{\rm CPT}^{\widehat P_k}(-h+c\eta_k)
=h^\alpha\underline G\!\left(\frac ch\right),
\end{equation}
and consequently
\begin{equation}\label{eq:standardized-dro-c-star}
\frac{c_{\rm DR}^*(h)}h
\in\arg\max_{0\leq t\leq c_{\max}/h}\underline G(t).
\end{equation}
If $\underline G$ is nondecreasing on $[0,t_{\rm DR}^*]$,
nonincreasing on $[t_{\rm DR}^*,\infty)$, and has the unique maximizer
$t_{\rm DR}^*$, then
\begin{equation}\label{eq:dro-c-star-scaling}
c_{\rm DR}^*(h)=\min\{h t_{\rm DR}^*,c_{\max}\}.
\end{equation}
\end{corollary}

\begin{proof}
Equation~\eqref{eq:cpt-positive-homogeneity} gives, for every $k$,
\[
V_{\rm CPT}^{\widehat P_k}(-h+c\eta_k)
=h^\alpha V_{\rm CPT}^{\widehat P_k}
\!\left(-1+\frac ch\eta_k\right).
\]
Taking the minimum over $k$ proves
Equation~\eqref{eq:dro-gap-homogeneity}. The substitution $t=c/h$ changes the
feasible interval to $[0,c_{\max}/h]$ and proves
Equation~\eqref{eq:standardized-dro-c-star}. Under the stated unimodality,
the maximizer of $\underline G$ over any truncated interval $[0,u]$ is
$\min\{t_{\rm DR}^*,u\}$. Taking $u=c_{\max}/h$ and multiplying by $h$
proves Equation~\eqref{eq:dro-c-star-scaling}.
\end{proof}

When $K=1$, $\underline G$ is the standardized objective associated with
Equation~\eqref{eq:c-star-problem}. If its unique unimodal maximizer is
$t_0^*$, Equations~\eqref{eq:standardized-dro-c-star} and
\eqref{eq:dro-c-star-scaling} reduce to
\begin{equation}\label{eq:standardized-c-star}
\frac{c^*(h)}h\in
\arg\max_{0\leq t\leq c_{\max}/h}V_{\rm CPT}(-1+t\eta)
\end{equation}
and
\begin{equation}\label{eq:c-star-gap-scaling}
c^*(h)=\min\{h t_0^*,c_{\max}\}.
\end{equation}
Thus, the nominal and robust results correspond to $K=1$ and $K>1$ in the
same homogeneity argument. If $\alpha_+\neq\alpha_-$, the gain and loss terms
contain different powers of $h$, and Equation~\eqref{eq:dro-c-star-scaling}
does not follow.

\subsection{Empirical CPT values and globally certified search}

The nominal problem evaluates one empirical CPT function, whereas the robust
problem evaluates the minimum of $K$ such functions. First consider $K=1$.
Generate $M$ values of $\eta$ from the fitted $Z$ and an independent standard
normal variable. For a fixed $c$, order
$y_s(c)=r_f-r_0+c\eta_s$ as
$y_{(1)}(c)\leq\cdots\leq y_{(M)}(c)$ and let
$k=\max\{i:y_{(i)}(c)<0\}$. Under equally weighted scenarios, the loss and
gain decision weights are
\begin{align}
\pi_i^-&=w^-\!\left(\frac iM\right)
-w^-\!\left(\frac{i-1}{M}\right),
&&i=1,\ldots,k,\label{eq:loss-decision-weights}\\
\pi_i^+&=w^+\!\left(\frac{M-i+1}{M}\right)
-w^+\!\left(\frac{M-i}{M}\right),
&&i=k+1,\ldots,M.\label{eq:gain-decision-weights}
\end{align}
The empirical CPT value is
\begin{equation}\label{eq:empirical-cpt-value}
\widehat V_M(c)=
\sum_{i=k+1}^{M}\pi_i^+y_{(i)}(c)^{\alpha_+}
-\lambda_{\ell}\sum_{i=1}^{k}\pi_i^-
[-y_{(i)}(c)]^{\alpha_-}.
\end{equation}
We first evaluate the objective on a regular grid over $[0,c_{\max}]$. A
bounded scalar optimization is applied to each neighboring interval that
contains a local maximum. The results are compared with both endpoints. This
calculation does not assume differentiability when a scenario crosses the
reference point.

For model $\widehat P_k$, order the generated scenarios as
$\eta_{k,(1)}\leq\cdots\leq\eta_{k,(M)}$. For $c>0$, the order of
$r_f-r_0+c\eta_{k,(i)}$ is fixed; only the gain-loss partition changes, at
the breakpoints
\begin{equation}\label{eq:empirical-cpt-breakpoints}
\mathcal B_{k,M}=
\left\{-\frac{r_f-r_0}{\eta_{k,(i)}}:
\eta_{k,(i)}\neq0,\quad
0<-\frac{r_f-r_0}{\eta_{k,(i)}}<c_{\max}\right\}.
\end{equation}
On an open interval containing no breakpoint, write
$y_{k,(i)}(c)=r_f-r_0+c\eta_{k,(i)}$. Termwise differentiation gives
\begin{align}
\widehat\Gamma_{k,M}'(c)
={}&\alpha_+\sum_{i:y_{k,(i)}(c)>0}
\pi_{k,i}^+\eta_{k,(i)}
[y_{k,(i)}(c)]^{\alpha_+-1}\notag\\
&+\lambda_{\ell}\alpha_-
\sum_{i:y_{k,(i)}(c)<0}
\pi_{k,i}^-\eta_{k,(i)}
[-y_{k,(i)}(c)]^{\alpha_--1}.
\label{eq:empirical-cpt-derivative}
\end{align}
The finite-scenario version of the candidate set is therefore
\begin{align}
\widehat{\mathcal C}_{\rm DR,M}
={}&\{0,c_{\max}\}
\cup\bigcup_{k=1}^K\mathcal B_{k,M}\notag\\
&\cup\bigcup_{k=1}^K
\left\{c\in(0,c_{\max})\setminus
\bigcup_{j=1}^K\mathcal B_{j,M}:
\widehat\Gamma_{k,M}'(c)=0\right\}\notag\\
&\cup\bigcup_{1\leq i<j\leq K}
\left\{c\in(0,c_{\max}):
\widehat\Gamma_{i,M}(c)=\widehat\Gamma_{j,M}(c)\right\}.
\label{eq:empirical-dro-candidate-set}
\end{align}
Stationary points and intersections are sought separately within the open
intervals determined by Equation~\eqref{eq:empirical-cpt-breakpoints}.
Evaluating the lower envelope on
$\widehat{\mathcal C}_{\rm DR,M}$ gives a feasible point. This calculation is
not a global certificate because a root-finding routine that brackets only
sign changes can miss an even-multiplicity tangency.

For a fixed model $k$, let $a_0=r_f-r_0$ and define the ranked scenario
contribution
\begin{equation}\label{eq:ranked-cpt-contribution}
h_{k,i}(y)=
\begin{cases}
\pi_{k,i}^{+}y^{\alpha_+},&y\geq0,\\
-\lambda_{\ell}\pi_{k,i}^{-}(-y)^{\alpha_-},&y<0.
\end{cases}
\end{equation}
This function is continuous and nondecreasing on $\R$, and
\begin{equation}\label{eq:branch-contribution-sum}
\widehat\Gamma_{k,M}(c)
=\sum_{i=1}^{M}h_{k,i}(a_0+c\eta_{k,(i)}).
\end{equation}
For a closed interval $I=[\ell,u]\subset[0,c_{\max}]$, define
\begin{align}
L_k(I)&=\sum_{i=1}^{M}\min\bigl\{
h_{k,i}(a_0+\ell\eta_{k,(i)}),
h_{k,i}(a_0+u\eta_{k,(i)})\bigr\},
\label{eq:branch-interval-lower}\\
U_k(I)&=\sum_{i=1}^{M}\max\bigl\{
h_{k,i}(a_0+\ell\eta_{k,(i)}),
h_{k,i}(a_0+u\eta_{k,(i)})\bigr\}.
\label{eq:branch-interval-upper}
\end{align}

\begin{proposition}[Interval certificate for a finite-scenario lower envelope]
\label{prop:empirical-dro-certificate}
For every $c\in I$ and $k=1,\ldots,K$,
\begin{equation}\label{eq:branch-interval-enclosure}
L_k(I)\leq\widehat\Gamma_{k,M}(c)\leq U_k(I).
\end{equation}
Consequently, if
\begin{equation}\label{eq:lower-envelope-interval-upper}
U(I)=\min_{1\leq k\leq K}U_k(I),
\end{equation}
then
\begin{equation}\label{eq:lower-envelope-supremum-bound}
\sup_{c\in I}\min_{1\leq k\leq K}
\widehat\Gamma_{k,M}(c)\leq U(I).
\end{equation}
Let $\mathcal P$ be a finite interval partition of $[0,c_{\max}]$, let
$\widehat c$ be any feasible point, and define
\begin{equation}\label{eq:global-certificate-bounds}
\underline v=\min_k\widehat\Gamma_{k,M}(\widehat c),
\qquad
\overline v=\max_{I\in\mathcal P}U(I).
\end{equation}
If every interval with upper bound at most $\underline v$ has been discarded
and the remaining intervals cover every point not yet excluded, the true
finite-scenario optimum $v_M^*$ satisfies
\begin{equation}\label{eq:global-certificate-gap}
\underline v\leq v_M^*\leq\max\{\underline v,\overline v\}.
\end{equation}
In particular, if
$\max\{\underline v,\overline v\}-\underline v\leq\varepsilon$, then
$\widehat c$ is an $\varepsilon$-global solution.
\end{proposition}

\begin{proof}
For fixed $i$, the affine map $a_0+c\eta_{k,(i)}$ attains its minimum and
maximum over $I$ at its endpoints. The derivative of
$h_{k,i}$ is nonnegative on either side of zero, and $h_{k,i}$ is continuous
at zero. The range of each scenario contribution is therefore enclosed by its
two endpoint values. Summing the enclosures proves
Equation~\eqref{eq:branch-interval-enclosure}. For every $c\in I$,
\[
\min_k\widehat\Gamma_{k,M}(c)
\leq\min_k U_k(I)=U(I),
\]
and taking the supremum over $I$ proves
Equation~\eqref{eq:lower-envelope-supremum-bound}.

The value $\underline v$ is attained at a feasible point, so it cannot exceed
$v_M^*$. The undiscarded intervals cover every point that could improve
$\underline v$, and the objective over each such interval does not exceed its
$U(I)$. Every discarded interval has upper bound at most $\underline v$.
Therefore $v_M^*$ is bounded above as in
Equation~\eqref{eq:global-certificate-gap}. The stopping condition yields
$0\leq v_M^*-\underline v\leq\varepsilon$.
\end{proof}

The implementation uses the candidate search in
Equation~\eqref{eq:empirical-dro-candidate-set} to obtain $\widehat c$, then
initializes $\mathcal P$ by equal-width intervals, repeatedly bisects the
interval with the largest $U(I)$, and updates the incumbent lower bound. The
certificate does not rely on sign changes at stationary points or branch
contacts. It certifies the discrete objective for the specified $M$ scenarios;
scenario-compression error and parameter-estimation error remain separate
statistical errors.

The predictive comparison determines $\widehat{\mathfrak P}$, and the
training sample determines $q_0$. Each $\widehat P_k$ then determines the law
of $\eta_k$ and the function $\widehat\Gamma_{k,M}$. The robust exposure is
obtained from their certified lower envelope.
 \section{Empirical design}\label{sec:empirical-design}

\subsection{Data and temporal split}

The marginal analysis uses daily log returns for AAPL, AMZN, GOOGL, and MSFT.
The multivariate analysis uses adjusted closing prices for 30 stocks downloaded
from Yahoo Finance. Both panels run from 2 January 2020 through 1 June 2026 and
contain 1,611 common trading days. Log differencing produces 1,610 observations
for each marginal series and for the 30-dimensional return vector. The tickers
are MMM, AXP, AMGN, AMZN, AAPL, BA,
CAT, CVX, CSCO, KO, DIS, GS, HD, HON, IBM, JNJ, JPM, MCD, MRK, MSFT, NKE,
NVDA, PG, CRM, SHW, TRV, UNH, VZ, V, and WMT. Returns are multiplied by 100,
so reported location, skewness-loading, and return quantities are in daily
percentage units.

All four marginal series have positive excess kurtosis. The sample skewness is
negative for AMZN, GOOGL, and MSFT and close to zero for AAPL. We retain the
skewness loading $\gamma$ in every fitted marginal model.

\begin{table}[htbp]
\centering
\caption{Descriptive statistics for daily log returns (percentage units).}
\label{tab:descriptive}
\begin{tabular}{lrrrrrr}
\toprule
Asset & Mean & Std. dev. & Skewness & Excess kurtosis & Minimum & Maximum\\
\midrule
AAPL  & 0.090 & 1.971 &  0.021 &  9.468 & -13.771 & 14.262\\
AMZN  & 0.063 & 2.227 & -0.074 &  7.230 & -15.140 & 12.695\\
GOOGL & 0.106 & 2.035 & -0.096 &  6.772 & -12.368 &  9.735\\
MSFT  & 0.069 & 1.875 & -0.246 & 10.689 & -15.945 & 13.293\\
\bottomrule
\end{tabular}
\end{table}

Both datasets are split chronologically, without random shuffling: the first
70\% is used for training and the remaining 30\% for evaluation. The marginal
and multivariate training samples each contain 1,127 observations, and their
holdout samples each contain 483 observations. A univariate Gaussian is the marginal
benchmark, and a multivariate Gaussian with an unrestricted covariance matrix
is the joint benchmark. Gamma mixing and the NPMLE are fitted from multiple
starting values; the other parametric fits are initialized from sample
moments.

\subsection{Predictive criteria and ambiguity-set construction}

The primary predictive criterion is the mean holdout log score
\begin{equation}\label{eq:log-score}
\overline{LS}=n_{\rm te}^{-1}
\sum_{i\in\mathcal I_{\rm te}}\log\widehat f(x_i),
\end{equation}
where a larger value indicates better density prediction. For a marginal
distribution, the Kolmogorov-Smirnov distance
\begin{equation}\label{eq:ks}
D_{\rm KS}=\sup_x|F_{\rm te}(x)-\widehat F(x)|
\end{equation}
is used as an additional global diagnostic. The holdout log score is the main
selection criterion. Tail quantile errors and $D_{\rm KS}$ are reported as
additional diagnostics.

A point ranking does not account for sampling variation in the
30-dimensional log-score differences.
Let $\ell_{k,t}$ be the log score of model $k$ at holdout date $t$, and let
$b$ index the model with the largest holdout mean. For the paired difference
\begin{equation}\label{eq:paired-logscore-difference}
d_{k,t}=\ell_{b,t}-\ell_{k,t},
\end{equation}
we use a circular block bootstrap with block length five and 4,000 replications.
Model $k$ is retained in Equation~\eqref{eq:finite-ambiguity-set} if the 95\%
bootstrap interval for $\E(d_{k,t})$ contains zero. We also calculate energy
scores, variogram scores, fixed-projection CRPS, and projection PIT
diagnostics. These additional scores are not used to modify the ambiguity set
after calculating the CPT decisions.

\subsection{CPT and numerical settings}

The annual risk-free rate is 1.25\%. An annual rate $a$ is converted to a
daily percentage return by
\begin{equation}\label{eq:annual-daily-rate}
r_{\rm day}(a)=100\left\{(1+a)^{1/252}-1\right\}.
\end{equation}
The gross risky-asset exposure is bounded by $L=1$. We use the benchmark CPT
parameters reported by Tversky and Kahneman:\cite{Tversky_Kaheman_1992_CPT}
\begin{equation}\label{eq:cpt-empirical-parameters}
\alpha_+=\alpha_-=0.88,
\qquad
\lambda_{\ell}=2.25,
\qquad
\delta_+=0.61,
\qquad
\delta_-=0.69.
\end{equation}
The annual reference return is set to 0\%, 5\%, or 10\%; a continuous
sensitivity calculation covers 0\%--15\%.

The common direction is constructed from the 30-dimensional training mean
$\overline X_{\rm tr}$ and covariance $S_{\rm tr}$:
\begin{equation}\label{eq:empirical-common-direction}
\overline v_{\rm tr}=\overline X_{\rm tr}-r_f\bm1,
\qquad
q_0=\frac{S_{\rm tr}^{-1}\overline v_{\rm tr}}
{\overline v_{\rm tr}^{\top}S_{\rm tr}^{-1}\overline v_{\rm tr}},
\qquad
c_{\max}=\frac1{\|q_0\|_1}.
\end{equation}
This direction is determined entirely by training-sample moments and does not
change with the fitted mixing family.

For each model in the ambiguity set, we transform the same two-dimensional
scrambled Sobol sequence into draws of $Z$ and $N$. We first generate $2^{17}$
projected scenarios and then retain 1,024 equally weighted quantile scenarios
for $\widehat\Gamma_{k,M}(c)$. The semianalytic search solves for the
reference-point breakpoints, branch stationary points, and pairwise branch
intersections in Equation~\eqref{eq:empirical-dro-candidate-set}. Its best
feasible point initializes the interval branch-and-bound procedure in
Proposition~\ref{prop:empirical-dro-certificate}. The certificate starts from
1,201 intervals and terminates when the difference between the global upper
and lower bounds is at most $10^{-10}$. Outward-rounding safeguards are added
to every interval sum and mean.

For any branch or robust method $\mathcal M$, let
$\widehat x_{\mathcal M}^*=\widehat c_{\mathcal M}^*q_0$. The empirical CPT
value over the holdout period is computed from
\begin{equation}\label{eq:test-cpt-outcome}
Y_{t,\mathcal M}^{\rm te}=r_f-r_0+
\widehat x_{\mathcal M}^{*\top}(X_t-r_f\bm1),
\qquad t\in\mathcal I_{\rm te},
\end{equation}
using Equations~\eqref{eq:loss-decision-weights}--\eqref{eq:empirical-cpt-value}.
The candidate-model parameters and $q_0$ are estimated from the training
sample. Since the holdout log scores determine membership in the ambiguity
set, CPT values calculated over the same holdout dates are descriptive rather
than an independent out-of-sample evaluation.
 \section{Empirical results}\label{sec:results}

\subsection{Marginal comparison of parametric mixing laws}

Table~\ref{tab:marginal-family} averages the results of the six parametric
mixing specifications over the four marginal return series. GIG gives the
largest training score, whereas lognormal mixing gives the largest holdout
score and the smallest AIC and BIC. Thus, the training score and the other
three criteria do not give the same ranking. We report inverse-gamma mixing in
the detailed comparison because it has the largest 30-dimensional holdout
score. The NPMLE is used as the semi-parametric comparator. The ambiguity set
in Table~\ref{tab:mv-ambiguity-set} is constructed from the multivariate
scores and does not depend on this choice of presentation.

\begin{table}[H]
\centering
\caption{Marginal comparison of six parametric mixing laws. Scores and runtime
are averaged over four assets; AIC and BIC are summed.}
\label{tab:marginal-family}
\begin{tabular}{lrrrrr}
\toprule
Mixing law & Train score & Holdout score & AIC & BIC & Time (s)\\
\midrule
GIG              & \textbf{-2.0848} & -1.9457 & 18836.9 & 18937.4 & 0.271\\
inverse Gaussian & -2.0855 & -1.9455 & 18835.0 & 18915.5 & 0.074\\
inverse gamma    & -2.0859 & -1.9456 & 18838.0 & 18918.5 & 0.064\\
gamma            & -2.0886 & -1.9503 & 18862.4 & 18942.8 & 0.314\\
exponential      & -2.0906 & -1.9507 & 18873.2 & 18933.5 & 0.014\\
lognormal        & -2.0853 & \textbf{-1.9452} & \textbf{18833.0} & \textbf{18913.5} & 0.073\\
\bottomrule
\end{tabular}
\end{table}

\begin{table}[H]
\centering
\footnotesize
\setlength{\tabcolsep}{2.8pt}
\caption{Marginal parameter estimates under inverse-gamma mixing and the NPMLE.}
\label{tab:parameters}
\begin{tabular}{llrrrrrrrr}
\toprule
Asset & Method & $\widehat\mu$ & $\widehat\gamma$ & $\widehat\sigma$ & $\widehat m$
& $\widehat\alpha$ & $\widehat\beta$ & $\widehat\nu$ & Supports\\
\midrule
\multirow{2}{*}{AAPL}  & inverse gamma & 0.146 & -0.052 & 2.058 & 0.977 & 2.050 & 1.025 & 4.100 & --\\
                       & NPMLE          & 0.078 &  0.017 & 2.058 & 1.001 & --    & --    & --    & 22\\
\multirow{2}{*}{AMZN}  & inverse gamma & 0.092 & -0.028 & 2.297 & 1.032 & 2.063 & 1.097 & 4.127 & --\\
                       & NPMLE          & 0.104 & -0.041 & 2.297 & 0.998 & --    & --    & --    & 21\\
\multirow{2}{*}{GOOGL} & inverse gamma & 0.234 & -0.148 & 2.080 & 0.988 & 2.050 & 1.038 & 4.100 & --\\
                       & NPMLE          & 0.274 & -0.188 & 2.080 & 0.991 & --    & --    & --    & 35\\
\multirow{2}{*}{MSFT}  & inverse gamma & 0.196 & -0.103 & 1.980 & 0.971 & 2.050 & 1.020 & 4.100 & --\\
                       & NPMLE          & 0.132 & -0.036 & 1.980 & 1.000 & --    & --    & --    & 34\\
\bottomrule
\end{tabular}
\end{table}

Table~\ref{tab:parameters} also reports the skewed-$t$ parameter
$\widehat\nu=2\widehat\alpha$. The inverse-gamma estimates for AAPL, GOOGL,
and MSFT are equal to the imposed lower bound $\widehat\alpha=2.05$ and are
therefore constrained estimates. The AMZN estimate is above the bound. For
AAPL, the estimated inverse-gamma loading is negative, whereas the NPMLE
loading is slightly positive. The other three NPMLE loadings are negative.

\begin{figure}[H]
\centering
\includegraphics[width=0.98\textwidth]{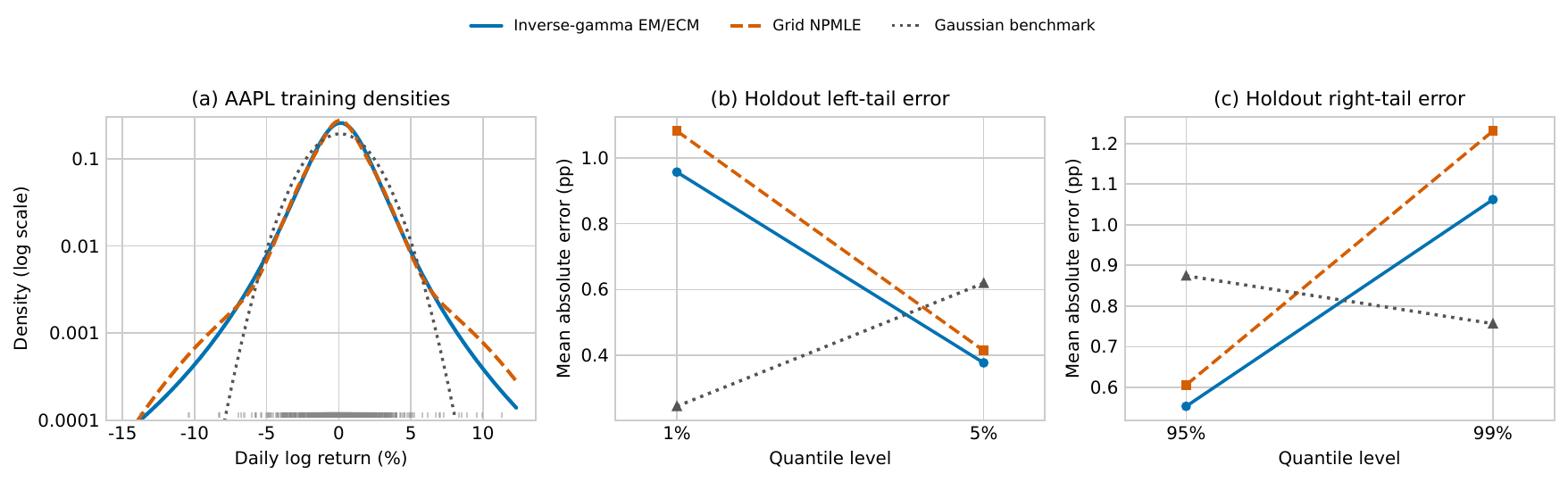}
\caption{Training log densities and holdout tail-quantile errors. Panel (a)
shows three fitted densities for AAPL on a logarithmic vertical scale. Panels
(b) and (c) report mean absolute quantile errors over the four assets for the
left and right tails separately.}
\label{fig:aapl-density}
\end{figure}

\begin{figure}[H]
\centering
\includegraphics[width=0.98\textwidth]{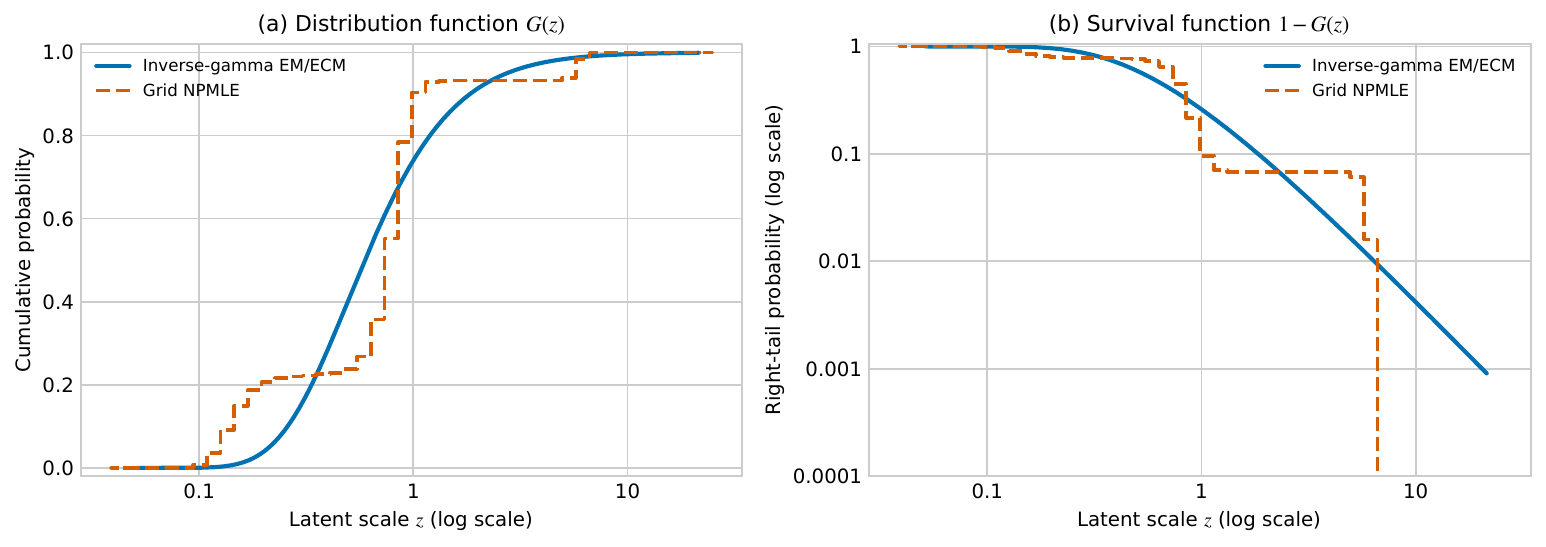}
\caption{Parametric and semi-parametric estimates of the AAPL latent-scale law.
The panels show the distribution and survival functions. Both horizontal axes
are logarithmic, and the survival probability is also plotted on a logarithmic
scale.}
\label{fig:aapl-mixing}
\end{figure}

In Figure~\ref{fig:aapl-density}(a), the inverse-gamma and NPMLE densities are
close over the central range and decay more slowly than the Gaussian density.
The tail panels give a less uniform ranking. Inverse gamma has the smallest
average error at the 5\% and 95\% quantiles, but its error exceeds that of the
Gaussian model at the 1\% left-tail quantile, and it is not the best model at
the 99\% quantile. Hence the holdout log score and the four tail-quantile
errors should be read as separate diagnostics.

Figure~\ref{fig:aapl-mixing} compares the estimated distribution and survival
functions of the latent scale. The inverse-gamma distribution has unbounded
support, whereas the fitted grid distribution has no mass above its largest
support point. The difference between the mixing distributions is much larger
than the difference between their fitted central return densities in
Figure~\ref{fig:aapl-density}(a).

\subsection{Marginal holdout log scores}

\begin{table}[H]
\centering
\caption{Mean holdout log score by asset.}
\label{tab:oos}
\begin{tabular}{lrrr}
\toprule
Asset & Inverse gamma & NPMLE & Gaussian\\
\midrule
AAPL  & -1.873 & \textbf{-1.865} & -2.001\\
AMZN  & \textbf{-2.079} & -2.081 & -2.150\\
GOOGL & \textbf{-2.020} & -2.025 & -2.079\\
MSFT  & -1.811 & \textbf{-1.807} & -1.930\\
\midrule
Mean  & -1.946 & \textbf{-1.944} & -2.040\\
\bottomrule
\end{tabular}
\end{table}

Inverse-gamma mixing has the larger holdout score for AMZN and GOOGL, and the
NPMLE has the larger score for AAPL and MSFT. The four-asset average is largest
for the NPMLE, followed by inverse gamma and the Gaussian model. All parameters
in this table are estimated from the training sample.

\subsection{Thirty-dimensional parametric and semi-parametric fits}

The joint models contain a location vector, a skewness-loading vector, and a
positive-definite conditional covariance matrix. The GIG estimate in
Table~\ref{tab:mv-family} has $\widehat\psi\approx0$, and its training and
holdout scores are almost equal to those of inverse-gamma mixing. The fitted
GIG model is therefore close to the inverse-gamma boundary. The holdout-score
difference is $1.9\times10^{-5}$. Inverse gamma uses one fewer parameter and
has the smaller AIC and BIC. Gamma and exponential mixing have the two lowest
holdout scores among the parametric mixtures.

\begin{table}[H]
\centering
\small
\setlength{\tabcolsep}{4pt}
\caption{Comparison of 30-dimensional parametric mixing laws.}
\label{tab:mv-family}
\begin{tabular}{lrrrrr}
\toprule
Mixing law & Train score & Holdout score & AIC & BIC & $\widehat m$\\
\midrule
GIG              & -48.10645 & -52.53147 & 109485.9 & 112135.3 & 1.008\\
inverse Gaussian & -48.13957 & -52.53853 & 109558.6 & 112203.0 & 1.001\\
inverse gamma    & -48.10645 & \textbf{-52.53145} & \textbf{109483.9} & \textbf{112128.3} & 1.007\\
gamma            & -48.22706 & -52.62142 & 109755.8 & 112400.2 & 0.978\\
exponential      & -48.33818 & -52.68996 & 110004.2 & 112643.6 & 1.019\\
lognormal        & -48.13285 & -52.54160 & 109543.4 & 112187.8 & 0.965\\
\bottomrule
\end{tabular}
\end{table}

\begin{table}[H]
\centering
\caption{Thirty-dimensional inverse-gamma, NPMLE, and Gaussian fits.}
\label{tab:mv-score}
\begin{tabular}{lrrrr}
\toprule
Method & Train score & Holdout score & $\widehat m$ & Effective supports\\
\midrule
inverse-gamma EM/ECM & -48.10645 & \textbf{-52.53145} & 1.007 & --\\
grid NPMLE            & \textbf{-48.09088} & -52.53980 & 1.025 & 24\\
multivariate Gaussian & -51.26574 & -56.83583 & -- & --\\
\bottomrule
\end{tabular}
\end{table}

For inverse-gamma mixing, $\widehat\alpha=2.572$ and
$\widehat\beta=1.583$. The NPMLE has the larger training score, whereas
inverse gamma has the larger holdout score. Their holdout-score difference is
small relative to the difference between either mixture model and the
multivariate Gaussian model. Table~\ref{tab:mv-vector-appendix} reports the
30-dimensional vector estimates.

\begin{table}[H]
\centering
\small
\setlength{\tabcolsep}{4pt}
\caption{Block-bootstrap log-score differences and the 30-dimensional ambiguity set.}
\label{tab:mv-ambiguity-set}
\begin{tabular}{lrrrrc}
\toprule
Model & Holdout score & $\overline d_k$ & 95\% lower & 95\% upper & Retained\\
\midrule
GIG              & -52.53147 & 0.00002 & -0.00016 & 0.00018 & yes\\
inverse Gaussian & -52.53853 & 0.00708 & -0.01774 & 0.03365 & yes\\
inverse gamma    & -52.53145 & 0       & 0        & 0       & yes\\
lognormal        & -52.54160 & 0.01015 & -0.01332 & 0.03661 & yes\\
NPMLE            & -52.53980 & 0.00835 & -0.00398 & 0.02023 & yes\\
\midrule
gamma            & -52.62142 & 0.08996 & 0.02959 & 0.15821 & no\\
exponential      & -52.68996 & 0.15851 & 0.10904 & 0.20505 & no\\
multivariate Gaussian & -56.83583 & 4.30438 & 3.07520 & 5.63994 & no\\
\bottomrule
\end{tabular}
\end{table}

In Table~\ref{tab:mv-ambiguity-set}, $\overline d_k$ is the sample mean in
Equation~\eqref{eq:paired-logscore-difference}. Inverse gamma has the largest
holdout score. The intervals for GIG, inverse Gaussian, lognormal, and the
NPMLE contain zero, and the resulting ambiguity set is
\begin{equation}\label{eq:empirical-ambiguity-set}
\widehat{\mathfrak P}=
\{\text{GIG},\text{inverse Gaussian},\text{inverse gamma},
\text{lognormal},\text{NPMLE}\}.
\end{equation}
The intervals for gamma, exponential, and the multivariate Gaussian model are
strictly positive, so these models are not included in
$\widehat{\mathfrak P}$.

\begin{table}[H]
\centering
\small
\setlength{\tabcolsep}{4pt}
\caption{Conditional covariance and latent-scale diagnostics.}
\label{tab:mv-covdiag}
\begin{tabular}{lrrrrrr}
\toprule
Method & $\lambda_{\min}$ & $\lambda_{\max}$ & Cond. no. & Mean corr. & $\widehat m$ & $\widehat{\Var}(Z)$\\
\midrule
inverse-gamma EM/ECM & 0.428 & 44.206 & 103.222 & 0.362 & 1.007 & 1.775\\
grid NPMLE           & 0.429 & 44.124 & 102.875 & 0.362 & 1.025 & 1.684\\
\bottomrule
\end{tabular}
\end{table}

The eigenvalues, condition numbers, and mean correlations in
Table~\ref{tab:mv-covdiag} are close under the two estimators.
Figure~\ref{fig:dow30-corr} reports the inverse-gamma correlation matrix and
the difference $R_{\rm NPMLE}-R_{\rm iG}$. The largest absolute off-diagonal
difference is 0.0049, and the root mean squared difference is 0.0015. Hence
the fitted conditional correlation matrices are also close, despite the
different specifications of $G$.

\begin{figure}[H]
\centering
\includegraphics[width=0.98\textwidth]{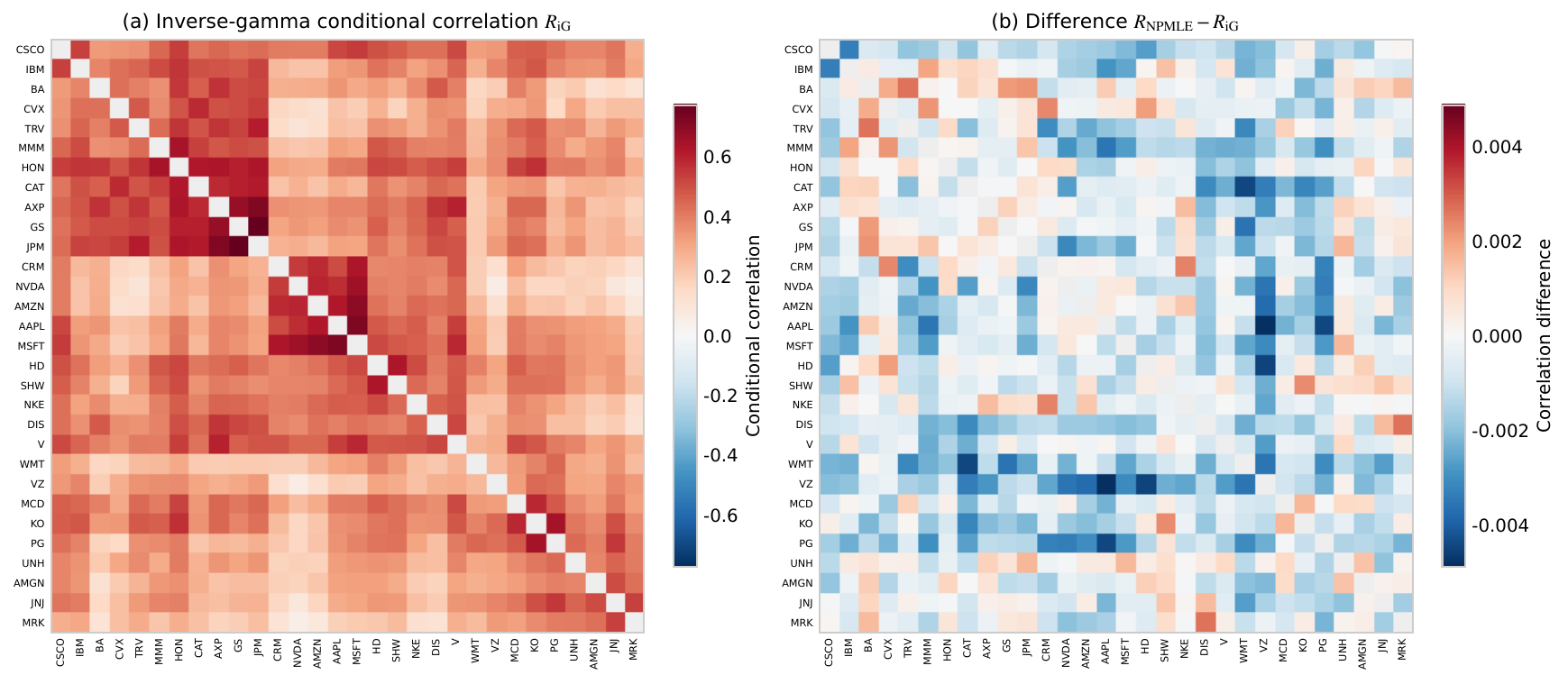}
\caption{Thirty-dimensional conditional correlation and estimation
difference. Assets are ordered by hierarchical clustering of the average
distance across the two estimated correlation matrices; diagonal entries are
left blank.}
\label{fig:dow30-corr}
\end{figure}

\subsection{Nominal and robust exposure on the common ray}

Table~\ref{tab:cpt-exposure} reports the solution for each retained model and
for the robust lower envelope. The direction $q_0$ and
$c_{\max}=0.045458$ are common to all rows. Since $L=1$,
$\|\widehat x^*\|_1=\widehat c^*/c_{\max}$. Gross risky-asset weights are
reported in percent; $c$ remains in daily percentage-return units.

\begin{table}[H]
\centering
\small
\setlength{\tabcolsep}{5pt}
\caption{Nominal and distributionally robust exposure on the common ray.}
\label{tab:cpt-exposure}
\begin{tabular}{lrrrrr}
\toprule
Method & $\widehat c^*(0\%)$ & $\widehat c^*(5\%)$
& Weight (5\%) & $\widehat c^*(10\%)$ & Weight (10\%)\\
\midrule
GIG              & 0 & 0.001964 & 4.320 & 0.004476 & 9.848\\
inverse Gaussian & 0 & 0.001942 & 4.272 & 0.004427 & 9.738\\
inverse gamma    & 0 & 0.001972 & 4.337 & 0.004477 & 9.849\\
lognormal        & 0 & 0.001977 & 4.349 & 0.004506 & 9.912\\
NPMLE            & 0 & 0.001948 & 4.286 & 0.004441 & 9.769\\
finite-set DRO   & 0 & \textbf{0.001948} & \textbf{4.286}
                       & \textbf{0.004441} & \textbf{9.769}\\
\bottomrule
\end{tabular}
\end{table}

For the 0\% annual reference return, every model gives $c^*=0$. For the 5\%
and 10\% references, the solutions are interior. The gross risky exposures are
approximately 4.3\% and 9.8\%, respectively. At both positive reference gaps,
the NPMLE is the unique active model in
Equation~\eqref{eq:active-worst-distribution}. Consequently, the robust and
NPMLE exposures are numerically equal. The interval
branch-and-bound gaps for the 0\%, 5\%, and 10\% references are
$9.20\times10^{-11}$, $1.00\times10^{-10}$, and
$9.99\times10^{-11}$, respectively.

At the 5\% annual reference, the model-based lower-envelope value at the robust
optimum is $-0.04478$, while the empirical holdout CPT value at the same
weights is $-0.04867$. At the 10\% reference, the corresponding values are
$-0.09245$ and $-0.10049$. The prospect values remain negative at both
positive reference returns. Since the holdout sample is used to construct the
ambiguity set, the empirical holdout values in this paragraph are not
independent performance estimates.

\begin{figure}[H]
\centering
\includegraphics[width=0.98\textwidth]{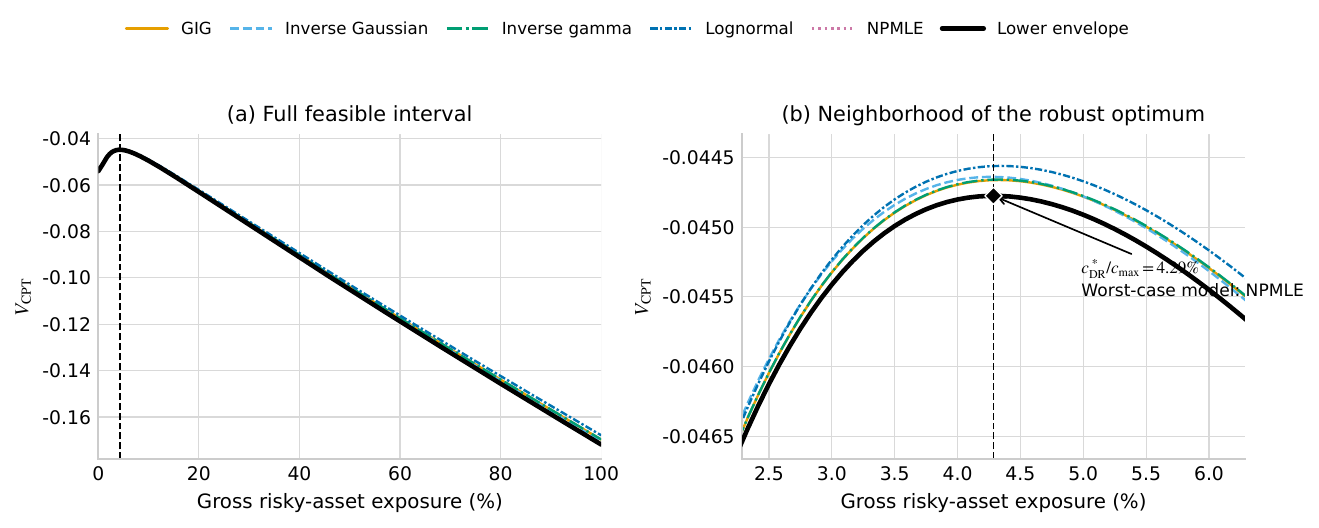}
\caption{Candidate-model CPT branches and their robust lower envelope at a 5\%
annual reference return. The left panel shows the full leverage interval; the
right panel enlarges the neighborhood of the robust optimum. The horizontal
axis is gross risky-asset exposure. The black curve is the pointwise minimum of
the five model-specific branches.}
\label{fig:cpt-objective}
\end{figure}

Figure~\ref{fig:cpt-objective} plots the five model-specific CPT functions and
their pointwise minimum. The functions are close over the feasible interval.
In a neighborhood of the optimum, the NPMLE function is the minimum and is
therefore the active function in the robust objective.
 \section{Discussion}\label{sec:discussion}

The marginal results give different rankings under the training score, the
holdout score, AIC, BIC, and the tail-quantile errors. In 30 dimensions, the
GIG and inverse-gamma likelihoods are almost equal. The block-bootstrap
intervals for inverse Gaussian, lognormal, and the NPMLE also contain zero when
they are compared with inverse gamma. We therefore retain these five models
rather than selecting inverse gamma only because it has the largest point
estimate of the holdout score.

The estimation and decision steps are summarized by
\begin{equation}\label{eq:estimation-decision-chain}
\{(\widehat G_k,\widehat\mu_k,\widehat\gamma_k,
\widehat\Sigma_k)\}_{k=1}^{K}
\longrightarrow q_0
\longrightarrow\{\widehat\eta_k,\widehat\Gamma_k(c)\}_{k=1}^{K}
\longrightarrow\min_{1\leq k\leq K}\widehat\Gamma_k(c)
\longrightarrow\widehat c_{\rm DR}^*.
\end{equation}
The direction $q_0$ is estimated once from the training mean and covariance.
Every fitted law is then evaluated at the portfolio $cq_0$. Since
$\widehat m_k$, $\widehat\mu_k$, $\widehat\gamma_k$, and
$\widehat\Sigma_k$ depend on $k$, the projected mean
$\E_{\widehat P_k}(\eta_k)$ is generally model dependent. The equality
$\E(\eta)=1$ applies only to a model's own normalized fund direction and is
not imposed on $Z$ or on the robust projections.

When $\alpha_+=\alpha_-$, the reference gap factors out of every CPT function
and their pointwise minimum. The mixing law remains in
$V_{\rm CPT}^{\widehat P_k}(-1+t\eta_k)$ and determines which function is
active. Equation~\eqref{eq:dro-c-star-scaling} requires the stated common
curvature and unimodality conditions. It can also fail when the exposure bound
truncates the standardized optimum.

The following limitations remain. The grid NPMLE depends on its specified
support range. The ambiguity set contains the fitted families only and is not
a Wasserstein ball or a moment class. Uncertainty in the estimated direction
$q_0$ is not included. A single common latent scale cannot separately describe
market, sector, and idiosyncratic scale shocks. The CPT parameters are fixed at
published experimental estimates. Finally, the interval certificate applies
to the compressed finite-scenario objective; it does not bound scenario
approximation or parameter-estimation error.

The holdout period determines the members of $\widehat{\mathfrak P}$ and is
also used for the descriptive CPT calculation. It is therefore not an
independent portfolio evaluation. Such an evaluation would require an
additional sample or a rolling design and would also need to account for
turnover and transaction costs.

\section{Conclusion}\label{sec:conclusion}

We estimated six parametric mixing models and a grid NPMLE under the same
determinant constraint. The mixing mean $m=\E(Z)$ was estimated for every
model. In the marginal analysis, the training score, holdout score, information
criteria, and tail-quantile errors did not give a common ranking. Three
inverse-gamma estimates reached the imposed lower bound and were treated as
constrained estimates.

In the 30-dimensional analysis, all retained mixture models had higher holdout
log scores than the multivariate Gaussian model. The paired block-bootstrap
comparison retained GIG, inverse Gaussian, inverse gamma, lognormal, and the
NPMLE. Gamma, exponential, and Gaussian specifications were excluded. Thus,
the predictive comparison gave a set of mixing laws rather than a unique law.

On the common direction, the multivariate problem reduces to the maximization
of a scalar lower envelope. The robust solution exists, and the breakpoints,
stationary points, and branch intersections give a candidate set. Interval
branch and bound certifies the finite-scenario optimum. In the application,
the model-specific exposures were close, and the NPMLE was the active model at
the robust solution for both positive reference gaps. These results concern
the specified common direction. The homogeneity result determines the scaling
with the reference gap, while the ordering of the fitted CPT functions
determines the active model.
 \appendix

\section{A sufficient condition for continuity of the robust objective}
\label{app:dro-continuity-proof}

For the weighting functions in Equation~\eqref{eq:cpt-weight-functions},
there exist constants $C_+,C_->0$, depending only on $\delta_+$ and
$\delta_-$, such that
\begin{equation}\label{eq:probability-weight-upper-bound}
w^+(p)\leq C_+p^{\delta_+},
\qquad
w^-(p)\leq C_-p^{\delta_-},
\qquad 0\leq p\leq1.
\end{equation}
Indeed, $w^+(p)/p^{\delta_+}$ and $w^-(p)/p^{\delta_-}$ are bounded over
$0<p\leq1$ and have finite limits as $p\downarrow0$.

Let $a_0=r_f-r_0$ and $Y_k(c)=a_0+c\eta_k$. Integration by parts in
Equation~\eqref{eq:cpt-functional} gives
\begin{align}
\Gamma_k(c)
={}&\alpha_+\int_0^\infty y^{\alpha_+-1}
w^+\!\left\{\Pr_{\widehat P_k}[Y_k(c)>y]\right\}\,\mathrm dy\notag\\
&-\lambda_{\ell}\alpha_-\int_0^\infty y^{\alpha_--1}
w^-\!\left\{\Pr_{\widehat P_k}[Y_k(c)<-y]\right\}\,\mathrm dy.
\label{eq:cpt-tail-integral}
\end{align}
Let $q_k$ satisfy Equation~\eqref{eq:dro-moment-condition}. For
$0\leq c\leq c_{\max}$ and $y\geq2|a_0|$,
\begin{align}
\Pr_{\widehat P_k}\{|Y_k(c)|>y\}
&\leq
\Pr_{\widehat P_k}\left\{|\eta_k|>
\frac{y-|a_0|}{c_{\max}}\right\}\notag\\
&\leq(2c_{\max})^{q_k}
\E_{\widehat P_k}|\eta_k|^{q_k}y^{-q_k},
\label{eq:uniform-tail-bound}
\end{align}
where the second line is Markov's inequality. Combining
Equations~\eqref{eq:probability-weight-upper-bound} and
\eqref{eq:uniform-tail-bound}, the two integrands in
Equation~\eqref{eq:cpt-tail-integral} are bounded at infinity by
\begin{equation}\label{eq:cpt-tail-dominators}
C_{k,+}y^{\alpha_+-1-q_k\delta_+}
\quad\text{and}\quad
C_{k,-}y^{\alpha_--1-q_k\delta_-},
\end{equation}
respectively. Both powers are integrable at infinity by
Equation~\eqref{eq:dro-moment-condition}. Near zero, the bounds
$0\leq w^\pm(p)\leq1$ give the integrable dominators
$y^{\alpha_+-1}$ and $y^{\alpha_--1}$. The two tail integrals are therefore
finite, with dominators that can be chosen uniformly over
$c\in[0,c_{\max}]$.

For $c>0$, positive-definite $\widehat\Sigma_k$ and $Z>0$ ensure that every
nonzero projection $Y_k(c)$ has a continuous distribution. If $c_n\to c$,
the tail probabilities in Equation~\eqref{eq:cpt-tail-integral} converge
pointwise except at a finite set of possible degenerate thresholds. At $c=0$,
the only possible discontinuity occurs at $|a_0|$, a Lebesgue-null threshold.
Dominated convergence therefore yields
\begin{equation}\label{eq:cpt-branch-continuity-limit}
\lim_{n\to\infty}\Gamma_k(c_n)=\Gamma_k(c).
\end{equation}
Thus $\Gamma_k$ is finite and continuous on $[0,c_{\max}]$, as required in
Proposition~\ref{prop:finite-distribution-dro}.

\clearpage
\section{Thirty-dimensional vector and mixing-law estimates}

Table~\ref{tab:mv-vector-appendix} reports the location, skewness loading, and
conditional standard deviation for inverse-gamma mixing and the NPMLE, where
$\widehat\sigma_j=\sqrt{\widehat\Sigma_{jj}}$. Location and skewness loadings
are in daily percentage-return units; conditional standard deviations are in
percentage units. A diagonal entry does not represent the full
$\widehat\Sigma$. The two $30\times30$ conditional covariance estimates are
provided with the reproducibility files.

\begin{table}[H]
\centering
\footnotesize
\renewcommand{\arraystretch}{0.86}
\setlength{\tabcolsep}{4pt}
\caption{Vector parameter estimates for the 30-dimensional NMVM.}
\label{tab:mv-vector-appendix}
\begin{tabular}{lrrr rrr}
\toprule
& \multicolumn{3}{c}{Inverse-gamma EM/ECM}
& \multicolumn{3}{c}{Grid NPMLE}\\
\cmidrule(lr){2-4}\cmidrule(lr){5-7}
Asset & $\widehat\mu_j$ & $\widehat\gamma_j$ & $\widehat\sigma_j$
& $\widehat\mu_j$ & $\widehat\gamma_j$ & $\widehat\sigma_j$\\
\midrule
MMM  &-0.017&-0.001&1.737&-0.027& 0.008&1.735\\
AXP  & 0.120&-0.060&2.178& 0.116&-0.055&2.175\\
AMGN & 0.034& 0.002&1.623& 0.034& 0.002&1.624\\
AMZN & 0.160&-0.094&2.355& 0.142&-0.077&2.355\\
AAPL & 0.175&-0.078&2.058& 0.170&-0.073&2.058\\
BA   & 0.155&-0.205&2.996& 0.145&-0.195&2.997\\
CAT  & 0.171&-0.091&2.069& 0.157&-0.077&2.070\\
CVX  & 0.049&-0.009&2.067& 0.037& 0.003&2.068\\
CSCO & 0.152&-0.140&1.590& 0.148&-0.136&1.590\\
KO   & 0.110&-0.082&1.208& 0.118&-0.089&1.207\\
DIS  & 0.005&-0.037&2.128& 0.012&-0.044&2.126\\
GS   & 0.137&-0.067&1.982& 0.124&-0.054&1.980\\
HD   & 0.181&-0.129&1.721& 0.164&-0.112&1.722\\
HON  & 0.105&-0.080&1.594& 0.096&-0.071&1.594\\
IBM  & 0.091&-0.045&1.541& 0.085&-0.039&1.541\\
JNJ  & 0.021&-0.009&1.185& 0.024&-0.012&1.184\\
JPM  & 0.143&-0.099&1.872& 0.134&-0.090&1.870\\
MCD  & 0.116&-0.082&1.254& 0.115&-0.082&1.255\\
MRK  & 0.053&-0.005&1.440& 0.059&-0.011&1.439\\
MSFT & 0.198&-0.101&1.909& 0.203&-0.105&1.907\\
NKE  & 0.117&-0.117&2.076& 0.116&-0.117&2.073\\
NVDA & 0.380&-0.107&3.575& 0.414&-0.140&3.574\\
PG   & 0.097&-0.059&1.199& 0.108&-0.070&1.198\\
CRM  & 0.237&-0.199&2.489& 0.223&-0.185&2.489\\
SHW  & 0.204&-0.157&1.830& 0.200&-0.153&1.831\\
TRV  & 0.129&-0.083&1.683& 0.125&-0.078&1.684\\
UNH  & 0.111&-0.059&1.693& 0.123&-0.071&1.696\\
VZ   &-0.005&-0.008&1.312& 0.006&-0.019&1.312\\
V    & 0.153&-0.116&1.615& 0.154&-0.116&1.614\\
WMT  & 0.105&-0.049&1.317& 0.110&-0.054&1.314\\
\bottomrule
\end{tabular}
\end{table}

For each model, write the projection on the common direction as
$\widehat\eta_k=\widehat a_k+\widehat b_kZ_k+
\widehat s_k\sqrt{Z_k}N$. Table~\ref{tab:cpt-projection-appendix} reports the
projection coefficients from Equation~\eqref{eq:robust-eta-means}. All five
models use the same $q_0$ and $c_{\max}=0.045458$, but
$\E_{\widehat P_k}(\widehat\eta_k)$ is not constrained to equal one. The last
column is the mean after compression to 1,024 equally weighted scenarios.
Asset-level entries of $q_0$ are stored in
\texttt{distributionally\_robust\_common\_direction.csv}.

\begin{table}[H]
\centering
\footnotesize
\setlength{\tabcolsep}{4pt}
\caption{Projection diagnostics on the common decision ray.}
\label{tab:cpt-projection-appendix}
\begin{tabular}{lrrrrrr}
\toprule
Model & $\widehat m_k$ & $\widehat a_k$ & $\widehat b_k$
& $\widehat s_k$ & $\E(\widehat\eta_k)$ & Scenario mean\\
\midrule
GIG              & 1.00792 & 0.90488 & 0.09296 & 7.25351 & 0.99858 & 0.99807\\
inverse Gaussian & 1.00063 & 0.88086 & 0.11906 & 7.24291 & 1.00000 & 1.00048\\
inverse gamma    & 1.00720 & 0.90531 & 0.09258 & 7.25339 & 0.99856 & 0.99805\\
lognormal        & 0.96471 & 0.87955 & 0.11985 & 7.24541 & 0.99517 & 0.99505\\
NPMLE            & 1.02503 & 0.96004 & 0.03898 & 7.25500 & 1.00000 & 0.99924\\
\bottomrule
\end{tabular}
\end{table}

\subsection{Mixing parameters and effective support points}

\begin{table}[H]
\centering
\footnotesize
\setlength{\tabcolsep}{3pt}
\caption{Mixing-law and computational results in 30 dimensions.}
\label{tab:mv-mixing-appendix}
\begin{tabular}{lcccccc}
\toprule
Method & Mixing parameters & $\widehat m$ & $\widehat{\Var}(Z)$
& Supports & Iterations & Time (s)\\
\midrule
inverse-gamma EM/ECM
& $\widehat\alpha=2.5716,\widehat\beta=1.5829$
& 1.0072 & 1.7748 & -- & 6 & 0.0318\\
grid NPMLE & 45-point initial grid
& 1.0250 & 1.6837 & 24 & 300 & 2.7368\\
\bottomrule
\end{tabular}
\end{table}

The grid routine reached its 300-iteration cap. Reapplying one complete
EM/ECM cycle at the saved estimate increased the average log likelihood by
$1.43\times10^{-6}$, so the reported scores and parameters are stable at the
precision used in the tables.

Table~\ref{tab:mv-support-appendix} lists the NPMLE support points whose
weights exceed $10^{-3}$. Their weights sum to 0.9982. The remaining 21 grid
points have total weight 0.0018 and are not counted as effective supports. All
support points and unrounded weights are supplied in the reproducibility files.

\begin{table}[H]
\centering
\small
\caption{Effective NPMLE support points and weights.}
\label{tab:mv-support-appendix}
\begin{tabular}{rrr@{\hspace{3em}}rrr}
\toprule
Index & Support $z_m$ & Weight $\widehat\pi_m$
& Index & Support $z_m$ & Weight $\widehat\pi_m$\\
\midrule
1  & 0.1757 & 0.0107 & 13 & 1.3730  & 0.0330\\
2  & 0.2035 & 0.0385 & 14 & 1.5902  & 0.0543\\
3  & 0.2357 & 0.0011 & 15 & 1.8418  & 0.0331\\
4  & 0.3662 & 0.0159 & 16 & 2.1331  & 0.0058\\
5  & 0.4241 & 0.1921 & 17 & 2.4705  & 0.0086\\
6  & 0.4912 & 0.1068 & 18 & 2.8613  & 0.0402\\
7  & 0.5689 & 0.0532 & 19 & 3.3139  & 0.0034\\
8  & 0.6589 & 0.0769 & 20 & 4.4453  & 0.0024\\
9  & 0.7631 & 0.1032 & 21 & 5.1484  & 0.0090\\
10 & 0.8838 & 0.0891 & 22 & 10.7290 & 0.0056\\
11 & 1.0236 & 0.0704 & 23 & 14.3917 & 0.0014\\
12 & 1.1855 & 0.0423 & 24 & 16.6682 & 0.0011\\
\bottomrule
\end{tabular}
\end{table}
 
\section*{Data availability statement}
The adjusted prices, log returns, fitted-model outputs, and numerical certificates used in this study are included with the reproducibility files accompanying the manuscript. The market data were obtained from Yahoo Finance; the retrieval script records the tickers and sample dates used in the analysis.

\section*{Code availability statement}
The Python implementations of the parametric EM/ECM estimators, the grid NPMLE, the block-bootstrap comparison, the CPT calculations, and the interval global-optimality certificate are included with the reproducibility files.

\section*{Conflict of interest}
The author declares no conflict of interest.

\section*{Funding}
This research received no external funding.

\begingroup
\small
\bibliographystyle{unsrtnat}

\begin{thebibliography}{16}
\providecommand{\natexlab}[1]{#1}
\providecommand{\url}[1]{\texttt{#1}}
\expandafter\ifx\csname urlstyle\endcsname\relax
  \providecommand{\doi}[1]{doi: #1}\else
  \providecommand{\doi}{doi: \begingroup \urlstyle{rm}\Url}\fi

\bibitem[Aas and Haff(2006)]{Aas_Haff_2006_GH_Skew_t}
Kjersti Aas and Ingrid~Hob{\ae}k Haff.
\newblock The generalized hyperbolic skew {Student's t}-distribution.
\newblock \emph{Journal of Financial Econometrics}, 4\penalty0 (2):\penalty0
  275--309, 2006.

\bibitem[Adcock et~al.(2015)Adcock, Eling, and
  Loperfido]{Adcock_2015_Skewed_in_Finance_Actuarial_science}
Christopher Adcock, Martin Eling, and Nicola Loperfido.
\newblock Skewed distributions in finance and actuarial science: A review.
\newblock \emph{The European Journal of Finance}, 21\penalty0
  (13--14):\penalty0 1253--1281, 2015.

\bibitem[Barndorff-Nielsen(1997)]{BarndorffNielsen1997NIG}
Ole~E. Barndorff-Nielsen.
\newblock Processes of normal inverse gaussian type.
\newblock \emph{Finance and Stochastics}, 2\penalty0 (1):\penalty0 41--68,
  1997.

\bibitem[Madan et~al.(1998)Madan, Carr, and Chang]{MadanCarrChang1998VG}
Dilip~B. Madan, Peter~P. Carr, and Eric~C. Chang.
\newblock The variance gamma process and option pricing.
\newblock \emph{Review of Finance}, 2\penalty0 (1):\penalty0 79--105, 1998.

\bibitem[Dempster et~al.(1977)Dempster, Laird, and
  Rubin]{Dempster_1977_EM_algorithm}
Arthur~P. Dempster, Nan~M. Laird, and Donald~B. Rubin.
\newblock Maximum likelihood from incomplete data via the {EM} algorithm.
\newblock \emph{Journal of the Royal Statistical Society: Series B
  (Methodological)}, 39\penalty0 (1):\penalty0 1--38, 1977.

\bibitem[Protassov(2004)]{Protassov2004EM}
R.~S. Protassov.
\newblock {EM}-based maximum likelihood parameter estimation for multivariate
  generalized hyperbolic distributions with fixed lambda.
\newblock \emph{Statistics and Computing}, 14\penalty0 (1):\penalty0 67--77,
  2004.

\bibitem[Kiefer and Wolfowitz(1956)]{KieferWolfowitz1956}
Jack Kiefer and Jacob Wolfowitz.
\newblock Consistency of the maximum likelihood estimator in the presence of
  infinitely many incidental parameters.
\newblock \emph{The Annals of Mathematical Statistics}, 27\penalty0
  (4):\penalty0 887--906, 1956.

\bibitem[Laird(1978)]{Laird1978NPMLE}
Nan Laird.
\newblock Nonparametric maximum likelihood estimation of a mixing distribution.
\newblock \emph{Journal of the American Statistical Association}, 73\penalty0
  (364):\penalty0 805--811, 1978.

\bibitem[Abudurexiti et~al.(2024{\natexlab{a}})Abudurexiti, He, Hu, Rachev,
  Sayit, and Sun]{Abudurexiti_2024_Mean_CVaR_Skewness_portfolio}
Nuerxiati Abudurexiti, Kai He, Dongdong Hu, Svetlozar~T. Rachev, Hasanjan
  Sayit, and Ruoyu Sun.
\newblock Portfolio analysis with mean-{CVaR} and mean-{CVaR}-skewness
  criteria based on mean-variance mixture models.
\newblock \emph{Annals of Operations Research}, 336\penalty0 (1):\penalty0
  945--966, 2024{\natexlab{a}}.
\newblock \doi{10.1007/s10479-023-05396-1}.

\bibitem[Tversky and Kahneman(1992)]{Tversky_Kaheman_1992_CPT}
Amos Tversky and Daniel Kahneman.
\newblock Advances in prospect theory: Cumulative representation of
  uncertainty.
\newblock \emph{Journal of Risk and Uncertainty}, 5\penalty0 (4):\penalty0
  297--323, 1992.

\bibitem[He and Zhou(2011)]{He_and_Zhou_2011_Singel_Period_Portfolio}
Xue~Dong He and Xun~Yu Zhou.
\newblock Portfolio choice under cumulative prospect theory: An analytical
  treatment.
\newblock \emph{Management Science}, 57\penalty0 (2):\penalty0 315--331, 2011.
\newblock \doi{10.1287/mnsc.1100.1269}.

\bibitem[Kwak and Pirvu(2018)]{Kwak_and_Pirvu_CPT_Skew_t}
Minsuk Kwak and Traian~A. Pirvu.
\newblock Cumulative prospect theory with generalized hyperbolic skewed $t$
  distribution.
\newblock \emph{SIAM Journal on Financial Mathematics}, 9\penalty0
  (1):\penalty0 54--89, 2018.

\bibitem[Consigli et~al.(2019)Consigli, Hitaj, and
  Mastrogiacomo]{Giorgio_2019_portfolio_CPT}
Giorgio Consigli, Asmerilda Hitaj, and Elisa Mastrogiacomo.
\newblock Portfolio choice under cumulative prospect theory: Sensitivity
  analysis and an empirical study.
\newblock \emph{Computational Management Science}, 16\penalty0 (1):\penalty0
  129--154, 2019.

\bibitem[Luxenberg et~al.(2024)Luxenberg, Schiele, and
  Boyd]{Eric_2024_CPT_convex_optimization}
Eric Luxenberg, Philipp Schiele, and Stephen Boyd.
\newblock Portfolio optimization with cumulative prospect theory utility via
  convex optimization.
\newblock \emph{Computational Economics}, 64\penalty0 (5):\penalty0 3027--3047,
  2024.

\bibitem[McNeil et~al.(2015)McNeil, Frey, and Embrechts]{McNeil2015QRM}
Alexander~J. McNeil, Ruediger Frey, and Paul Embrechts.
\newblock \emph{Quantitative Risk Management: Concepts, Techniques and Tools}.
\newblock Princeton University Press, revised edition, 2015.

\bibitem[Abudurexiti et~al.(2024{\natexlab{b}})Abudurexiti, Bayraktar, Hayashi,
  and Sayit]{Abudurexiti_2024_expected_utility}
Nuerxiati Abudurexiti, Erhan Bayraktar, Takaki Hayashi, and Hasanjan Sayit.
\newblock Two-fund separation under hyperbolically distributed returns and
  concave utility functions.
\newblock \emph{arXiv preprint arXiv:2410.04459}, 2024{\natexlab{b}}.
\newblock \doi{10.48550/arXiv.2410.04459}.
\newblock Revised February 2026.

\end{thebibliography}

\endgroup

\end{document}